\documentclass[letterpaper,11pt]{article}

\usepackage[margin=1in]{geometry}
\usepackage{amsmath,amsthm,amsfonts,amssymb,bm,bbm}
\usepackage{graphicx}
\graphicspath{ {./images/} }

\usepackage[square]{natbib}
\renewcommand{\cite}{\citep}

\usepackage{xcolor,colortbl}
\usepackage{mathtools}
\usepackage{multirow}
\usepackage[hidelinks]{hyperref}
\hypersetup{
    colorlinks,
    linkcolor={blue},
    citecolor={blue},
    urlcolor={blue}
}
\usepackage{cleveref}

\usepackage{tikz}
\usepackage{forest}
\usetikzlibrary{arrows.meta,tikzmark} 
\usepackage{thm-restate}

\usetikzlibrary{shapes,decorations,arrows,calc,arrows.meta,fit,positioning}
\tikzset{
    -Latex,auto,node distance =1 cm and 1 cm,semithick,
    state/.style ={ellipse, draw, minimum width = 0.7 cm},
    point/.style = {circle, draw, inner sep=0.04cm,fill,node contents={}},
    bidirected/.style={Latex-Latex,dashed},
    el/.style = {inner sep=2pt, align=left, sloped}
}

\newcommand{\xE}{\mathsf{E}}

\newcommand{\xEN}{\widetilde{\mathsf{E}}}
\newcommand{\xS}{\mathsf{S}}

\newcommand{\ES}{\xE\xS}
\newcommand{\SES}{\xS\xE\xS}
\newcommand{\EnSES}{\xEN\xS\xE\xS}

\newcommand{\bZ}{\mathbb{Z}}

\newcommand{\given}{\,|\,}

\DeclareMathOperator{\Exp}{\mathbb{E}}
\DeclareMathOperator{\Pro}{\mathbb{P}}

\DeclareMathOperator*{\argmax}{arg\,max}

\Crefname{equation}{}{}

\newenvironment{talign}
 {\align}
 {\endalign}
\newenvironment{talign*}
 {\csname align*\endcsname}
 {\endalign}

\newcolumntype{L}{>{$}l<{$}}
\newcolumntype{R}{>{$}r<{$}}
\newcolumntype{C}{>{$}c<{$}}
\newcommand{\myheader}[1]{\multicolumn{1}{C}{#1}}
\newcommand{\mytablecaption}[1]{\myheader{\rule{0pt}{5mm}#1}}

\newtheorem{theorem}{Theorem}
\newtheorem{lemma}[theorem]{Lemma}

\newtheorem{example}{Example}[section]

\newcommand{\varcolor}[1]{{\color{blue} #1}}

\title{Persuading a Credible Agent}
\author{Jiarui Gan, Abheek Ghosh, Nicholas Teh}

\date{%
University of Oxford, UK}
\begin{document}

\maketitle

\begin{abstract}
How to optimally persuade an agent who has a private type? When \emph{elicitation} is feasible, this amounts to a fairly standard principal-agent-style mechanism design problem, where the persuader employs a mechanism to first {elicit} the agent's type and then plays the corresponding persuasion strategy based on the agent's report. The optimal mechanism design problem in this setting is relatively well-understood in the literature, with incentive compatible (IC) mechanisms known to be optimal and computationally tractable. In this paper, we study this problem given a \emph{credible} agent, i.e., if the agent claims they are of a certain type in response to the mechanism's elicitation, then they will act optimally with respect to the claimed type, even if they are actually not of that type.

We present several interesting findings in this new setting that differ significantly from results in the non-credible setting. In terms of the structure of optimal mechanisms, we show that not only may IC mechanisms fail to be optimal, but all mechanisms following the standard `eliciting-then-persuading' mechanism design structure may be suboptimal. To achieve optimality requires two additional instruments---\emph{pre-signaling} and \emph{non-binding elicitation}---which naturally result in multi-stage mechanisms. 
We characterize optimal mechanisms under these design choices.
Based on our characterization, we provide a polynomial-time algorithm for computing optimal multi-stage mechanisms.
We also discover that in scenarios that allow for it, {\em partial} information elicitation can be employed to improve the principal's payoff even further. Though, surprisingly, an unbounded number of rounds of information exchange between the principal and the agent may be necessary to achieve optimality.
\end{abstract}

\section{Introduction}
Information design (or \emph{Bayesian persuasion}) models scenarios where a principal strategically reveals her private information to incentivize an agent to play preferred actions. 
For instance, a seller may want to reveal only partial details about her products to incentivize an agent to buy a particular one, or a school may want to issue only coarse grades to improve its students’ job outcomes~\cite{boleslavsky2015grading,kamenica2019bpsurvey,ostrovsky2010unraveling}.

The standard information design model assumes that the principal knows the utility function of the agent~\cite{kamenica2011bp}, which may not hold in practice. 
Numerous works subsequently study a model where the principal does not know the agent's type but has a prior over it \cite{alonso2018persuasionexperts,bernasconi2023onlineBP,castiglioni2020online,castiglioni2022type,castiglioni2023regretonlinebp,gill2008seq,kolotilin2017persuasion,perez2014interim}.
The principal may try to design the signaling scheme (i.e., strategically reveal her private information) to maximize her expected utility given only this prior. Another, possibly better, alternative for the principal, if feasible, is to ask the agent to reveal his type before the principal sends her signal, through an \emph{elicitation} process.
In this case, the agent may be \textit{credible}, i.e., take actions consistent with his reported type, or be \textit{non-credible}, i.e., report any type and then play any, possibly inconsistent, action afterward (essentially, cheap talk). 
In this paper, we study such scenarios, with a particular focus on credible agents.

Even though credibility has been relatively understudied in the literature, it remains crucial in many situations. One such scenario is when the principal and the agent both enjoy and want to protect some form of reputation. This may happen, for example, when the two parties are known businesses. The principal may still have private information about the state and the agent about his private type, but the agent may stick with the contract (type) signed by both the parties at an earlier stage of the game. Besides reputation, the agent may also want to stick to the previously settled type (contract) to avoid possible legal ramifications. 
From a behavioral game-theoretic perspective, agents may also be averse to playing actions that are inconsistent with their report due to psychological and social reasons that cause bounded rationality. 
In strategic scenarios, credibility may also arise voluntarily when an agent intends to behave consistently to change the principal's belief about their behavior. Such phenomenon have been termed {\em impersonation} or {\em imitative deception} in various studies \cite{gan2019imitative,kash2010impersonation}. 
Our work is an attempt to capture these situations and understand the computational and algorithmic aspects of mechanism design against credible agents.

We aim to compute payoff-maximizing mechanisms for the principal. In pursuit of this, we identify several crucial design aspects: whether the principal uses only incentive-compatible mechanisms; whether the principal sends other signals in addition to the action recommendations; whether the agent reports other information in addition to his credible type report; and, in general, how many rounds of communication between the principal and the agent are necessary? 
As it turns out, the presence of credibility changes the design of optimal mechanisms significantly. 
It necessitates additional design choices that would be unnecessary when the agent is non-credible. 
We introduce novel stages in the mechanism design and provide non-trivial examples to demonstrate strict payoff improvements they achieve. 
Moreover, we characterize optimal mechanisms under these additional design choices and prove the general optimaility they lead to. The characterization also result in an efficient algorithm to compute optimal mechanisms, thereby highlighting both the computational and utility benefits they offer.
Additionally, we provide matching hardness results that emerge when these additional design choices are not employed, thereby highlighting their computational benefits.

\subsection{Our Results}

In \Cref{sec:non-cred}, we begin with a non-credible agent and observe that incentive compatible (IC) policies are without loss of optimality. Given this, we can compute the optimal policy in polynomial time using a linear programming (LP) formulation. 
For a credible agent, we show that if the principal restricts herself to using only IC policies, then again she can compute the optimal IC policy in polynomial time using an LP. 
However, optimal IC policies may not be optimal overall, in fact, can be arbitrarily worse; we provide an example that demonstrates this. 
We then show that computing an optimal policy exactly or approximately is NP-hard using a tight reduction from {\sc Maximum Independent Set}. 

In Section~\ref{sec:multi_stage}, we go beyond two-stage interactions (of elicitation and signaling) 
and consider cases where the principal and the agent may interact over multiple rounds. 
For a non-credible agent, such interactions beyond two stages are redundant. 
However, with credible agents, we show that two additional stages---\emph{pre-signaling} and \emph{non-binding elicitation}---are necessary for achieving optimality.
We also show via non-trivial examples that these additional stages can strictly improve the principal's utility. 
Surprisingly, these additional stages also make the computation of optimal mechanisms tractable.
To derive this result, we characterize the class of optimal mechanisms and prove that a four-stage mechanism structure---involving (in order) pre-signaling, non-binding elicitation, binding type elicitation, and signaling---results in IC across multiple stages and is w.l.o.g. optimal against a credible agent. 
The characterization can be viewed as a revelation principle, where non-binding elicitation indeed reveals the agent's true type, despite them eventually imitating another agent type.
Based on the simple fixed structure from the characterization, we are able to formulate the computation of optimal mechanisms as an LP, which yields a polynomial-time algorithm. 

Finally, in Section \ref{sec:partial-info-elicit}, we observe that a specific form of elicitation which we call {\em partial information elicitation} can be used to further improve the principal's utility. We show that under these information elicitation structures, an unbounded number of rounds of information exchange between the principal and the agent may be necessary to achieve optimality. We prove a tight $\Theta(n)$ bound on the depth of optimal mechanisms, when nature has $n$ possible states.

Our results highlight sharp contrasts between the computational and algorithmic aspects of credible versus non-credible cases.

\subsection{Related Work}

We build on the standard Bayesian persuasion model of \citet{kamenica2011bp}, which was proposed to study the effects of (the provision of) information on the outcomes of strategic interactions.
The model has attracted widespread interest in both economics and theoretical computer science
(e.g., \citep{babichenko2017algorithmic,dughmi2016algorithmic,feng2024rationality,kolotilin2017persuasion,xu2020tractability}; also refer to the surveys by \citet{dughmi2017infostructuresurvey} and \citet{kamenica2019bpsurvey}).

Models with private agent type but without credibility has been well-motivated in the literature, in various forms of principal-agent problems (e.g., \citep{conitzer2002complexity,myerson1982optimal}). This is in general related to the economics concept of cheap talk \citep{farrell1996cheap}. 
Several works look into the Bayesian persuasion model where agents have a private type and reports it (possibly strategically) to the principal.
\citet{castiglioni2022type} look into computational problems in a model where the agent is asked to report their type to principal after the principal have committed to a menu of signaling schemes.
Subsequent work by \citet{bernasconi2023onlineBP} look into further computational questions in this model.
\citet{kolotilin2017persuasion} consider a model where the agent only needs to choose between two actions, and both the principal's and agent's utilities are linearly related to the state and the agent's type.
They ask the question of whether the principal can benefit by designing a complex mechanism under these assumptions.
Other works that consider models of Bayesian persuasion where the agent has a private type include \citep{gill2008seq,perez2014interim} and \citep{alonso2018persuasionexperts}.
In \emph{online} Bayesian persuasion \cite{castiglioni2020online,castiglioni2023regretonlinebp}, the agent has an unknown private type, and this type is chosen adversarially at each round, over multiple principal-agent interactions.
Such two-way information exchanges may also take place in sequential interactions as studied more recently in \cite{gan2023sequential}.

Another related line of work is strategic communication with costly lying \cite{kartik2008lies,kartik2009strategiclying,kephart2015complexity,kephart2016revelation,nguyen2021costly} or constraints on how they can lie \cite{yu2011mechanism,zhang2021automated}, which is somewhat the opposite of the \emph{cheap talk} model \cite{crawford1982cheaptalk}.
However, these works focus mostly on direct cost or constraints imposed on the principal's signaling or the agent's reporting. The credibility constraint in our model can be viewed as an implicit cost that depends on both the agent's report and their action: an infinitely large cost is incurred if they act in ways different from their report.

Our model of a credible agent is related to various studies on {\em impersonation} or {\em imitative deception} in algorithmic game theory \cite{birmpas2021optimally,gan2019manipulating,gan2019imitative,kash2010impersonation,nguyen2019imitative}. 
In these studies, an agent imitates the best responses of another agent type to change the principal's belief about their incentives. Such attempts can be effective especially when the principal learns the agent's incentives by querying the agent's best responses \cite{blum2015learning,peng2019learning}. 
From the principal's perspective, to maximize the utility of her learning outcome amounts to the same mechanism problem we study in this paper.
\citet{gan2019imitative} considered this mechanism design problem in normal-form Stackelberg games. 
In their model, the agent (follower) responds to the principal's (leader's) strategy according to the utility function they report to the principal, so as to trick the principal into permanently believing that they are of the reported type. 
Despite the similarity, they did not consider the broader set of approaches we study in this paper. Indeed, a key difference is that in their model, the principal does not have any information advantages over the agent (there is no hidden state of nature, and the principal only determines an action to play). Without such advantages, the standard mechanism structure, which maps the type the agent reports to an outcome, suffices for achieving optimality.

There are other recent works focusing on cheap-talk style conversations in Bayesian settings, where multiple agents with their own private information engage in conversations, through which they strategically reveal their private information
\cite{mao2022bilateral,leme2023bayesianconversations,arieli2023mediated}.
Such conversations may proceed in multiple rounds, similarly to multiple stages of information revelation and elicitation in the mechanism we design, but for different fundamental reasons. Crucially, all information exchanges in these models are cheap, in the sense that they do not directly change the players' payoffs, whereas in ours the agent's reporting is costly because of credibility: the agent will suffer an infinite cost if he behaves differently from the type he reports.

\section{Model and Preliminaries}\label{sec:prelim}
For a set $S$, let $\Delta(S)$ denote the set of all probability distributions over $S$. 
Let $[n] = \{1, 2, \ldots, n\}$ for a positive integer $n \in \bZ_{> 0}$.
In the standard information design model \citep{kamenica2011bp}, a principal (sender) wants to persuade an agent (receiver) to take an action that is desirable to the principal. There is a set of actions $A$, and the payoff of each action $a \in A$ to both the principal and the agent is determined by the state of nature $\theta \in \Theta$---we use $v(\theta, a)$ and $u(\theta, a)$ to denote those payoffs, respectively. We assume $\theta$ is drawn from a common prior distribution $\mu$, and the principal must commit to a signaling scheme $\pi(\cdot\given\theta) \in \Delta(G)$, where $G$ denotes some set of signals. For this standard model, the revelation principle tells that, without loss of generality, the set of signals corresponds to the set of actions, i.e., $G = A$.
The goal we adopt from the perspective of a principal is to find an optimal signaling strategy, which maximizes the principal's expected payoff for the outcome induced. 

\subsection{Agent with Private Type}
\label{sec:prelim:type}

In many scenarios, the agent has a private utility function unknown to the principal. 
The standard model can be extended to these scenarios as follows.
Let $T$ be a finite set of types for the agent.
An agent of type $t \in T$ has a utility function $u_t(\theta, a)$.
The probability that the agent is of type $t \in T$ is $\rho(t)$.
The distribution $\rho \in \Delta(T)$ is known to the principal. The utility function of the principal is $v(\theta, a)$, same as the standard model. 

When the principal cannot elicit any additional information from the agent about his type, she has to design her signaling scheme $\pi(\cdot\given\theta) \in \Delta(G)$, for each state $\theta$, where $G$ denotes the set of signals, given only the knowledge about the distribution of types $\rho$. 
Let $\pi(\theta, g) = \pi(g \given \theta) \cdot \mu(\theta)$.
We assume that the agent plays optimally given $\pi$. Formally, given $\pi$, we assume that an agent of type $t$ plays an action $a$ with probability $\alpha_{\pi,t}(a \given g)$ after getting a signal $g \in G$ from the principal; where $\alpha_{\pi,t}(a \given g) > 0$ only if playing $a$ maximizes agent's expected utility given $g$ and $\pi$. The expected utility of the agent with type $t$ is given by
\begin{equation}\label{eq:u:agent}
    \Exp_{(\theta, g) \sim \pi,\, a \sim \alpha_{\pi,t}(\cdot \given g)}[u_t(\theta, a)] = \sum_{\theta, g} \pi(\theta, g) \sum_a \alpha_{\pi, t} (a \given g) \cdot u_t(\theta, a).
\end{equation}
The expected utility of the principal is given by
\begin{equation}\label{eq:u:prin}
    \Exp_{t \sim \rho, (\theta, g) \sim \pi,\, a \sim \alpha_{\pi,t}(\cdot \given g)}[v(\theta, a)] = \sum_t \rho(t) \sum_{\theta, g} \pi(\theta, g) \sum_a \alpha_{\pi, t} (a \given g) \cdot v(\theta, a).
\end{equation}
In this model, where the principal cannot elicit any additional information from the agent, \citet{castiglioni2020online} and \citet{xu2020tractability} prove that the problem of finding a signaling scheme $\pi$ that maximizes the principal's utility is NP-hard to solve exactly or approximately. In particular, they show that for any constant $0 \le \epsilon < 1$, the principal cannot find a policy that gives her an expected reward of at least $(1 - \epsilon)$ times the optimal reward in polynomial time unless $P = NP$.

\subsection{Type Elicitation}

In this paper, we focus on the scenario with elicitation, i.e., when the principal can ask the agent for his type before sending a signal. The agent of type $t$ reports his type as $t' \in T$. The principal's signaling scheme $\pi_{t'}(\cdot\given\theta)$ is a function of both $\theta$ and $t'$. Let $\pi_{t'}(\theta, g) = \pi_{t'}(g \given\theta) \cdot \mu(\theta)$.
\Cref{fig:timeline} illustrates the timeline of the process.

We consider both {\em non-credible} and {\em credible} reporting, focusing primarily on the latter.
In the non-credible case, the agent can report any type and then play any action afterward---the action need not be consistent with the reported type. 
In contrast, in the credible case, the agent needs to behave as if they are the reported type, i.e., if the agent reports a type $t'$, then he must play an action that is optimal for the type $t'$ given his posterior belief (as a function of the common prior and the signal he receives from the principal). 
We formally introduce the pre-signaling, non-binding elicitation, and partial elicitation in Sections~\ref{sec:multi_stage} and \ref{sec:partial-info-elicit}.

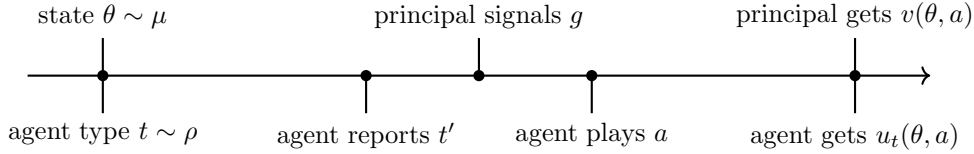
\begin{figure}
    \centering
    \small
    \begin{tikzpicture}
      \draw[->,thick] (-1,0) -- (11, 0);
      \draw[-,thick] (0,-0.5) node[below, align=center]{agent type $t \sim \rho$} -- (0,0) node[circle,fill,inner sep=1.5pt]{};
      \draw[-,thick] (0, 0.5) node[above, align=center]{state $\theta \sim \mu$} -- (0,0);
      \draw[-,thick] (3.5,-0.5) node[below, align=center]{agent reports $t'$} -- (3.5,0) node[circle,fill,inner sep=1.5pt]{};
      \draw[-,thick] (5, 0.5) node[above, align=center]{principal signals $g$} -- (5,0) node[circle,fill,inner sep=1.5pt]{};
      \draw[-,thick] (6.5,-0.5) node[below, align=center]{agent plays $a$} -- (6.5,0) node[circle,fill,inner sep=1.5pt]{};
      \draw[-,thick] (10,0.5) node[above, align=center]{principal gets $v(\theta,a)$} -- (10,0) node[circle,fill,inner sep=1.5pt]{};
      \draw[-,thick] (10,-0.5) node[below, align=center]{agent gets $u_t(\theta,a)$} -- (10,0);
    \end{tikzpicture}
    \caption{Timeline.}
    \label{fig:timeline}
\end{figure}

\section{Non-credible vs. Credible Reporting}
\label{sec:non-cred}

Let us first compare credible and non-credible reporting and present some algorithmic results.

\subsection{Non-credible Reporting}

When the agent is non-credible, we show that the principal can compute the optimal mechanism in polynomial time using an LP formulation.

First, notice that, when the agent is non-credible, IC is w.l.o.g.~using standard arguments: Any mechanism $\pi$ of the principal will induce an optimal reporting strategy of the agent given $\pi$, which can be encapsulated in a meta-mechanism that simulates the behavior of the agent given his true type. This meta-mechanism plays optimally for the agent and, therefore, is IC. It also provides the same expected reward to the principal as the original mechanism $\pi$. So, the principal can optimize over IC signaling policies w.l.o.g. 
Let $\pi_t$ denote the mechanism used by the principal when an agent reports type $t$. Let $\pi_t(\theta, a) = \pi_t(a \given \theta) \cdot \mu(\theta)$ denote the marginal probability of $(\theta, a)$.
The following LP computes the optimal signaling mechanism for the principal.
\begin{align}
\max_{(\pi_t)_{t \in T}} \quad
& \sum_{t \in T} \rho(t) \sum_{\theta, a} \pi_t(\theta, a) \cdot v(\theta, a) \label{eq:non-cred:1}
\end{align}
\begin{align}
\text{subject to} \quad
& 
\sum_{\theta, a} \pi_t(\theta, a) \cdot u_t(\theta, a) \ge 
\sum_{\theta, a} \pi_{t'}(\theta, a) \cdot u_t(\theta, f(a)), 
& \forall t, t' \in T,\ f: A \to A, 
\label{eq:non-cred:2} \\
& 
\sum_{a} \pi_t(\theta, a) = \mu(\theta),  
\qquad\qquad\qquad \ \forall t \in T, \theta \in \Theta, \\
&
\pi_t \in \Delta(\Theta \times A),
\qquad\qquad\qquad\qquad \forall t \in T.
\end{align} 
Specifically, the objective function captures the principal's expected payoff under the agent's truthful response. Constraint~\eqref{eq:non-cred:2} ensures that this truthful behavior is indeed incentivized, where: the left side is the agent's expected payoff for truth reporting, conditioned on type $t$; and the right side describes the payoff the agent would obtain if he reports a different type $t'$ and plays $f(a)$ when he gets the recommendation to play $a$.
This constraint also implicitly captures the following:
\begin{talign*}
    &\sum_{\theta} \pi_t(\theta, a) \cdot u_t(\theta, a) \ge \sum_{\theta} \pi_t(\theta, a) \cdot u_t(\theta, a'), &\forall t \in T,\, a, a' \in A,
\end{talign*}
so in addition to truthful type reporting $\pi$ incentivizes, every action recommendation of $\pi_t$ is also IC.
Since the constraint \eqref{eq:non-cred:2} is defined for all $f: A \to A$, it includes exponentially many linear constraints. 
However, since $f$ is not a function of the state $\theta$, the constraint can be equivalently written as follows:
\begin{align}
&\sum_{\theta, a} \pi_t(\theta, a) \cdot u_t(\theta, a) \ge 
\sum_{a}\ 
\underbrace{\max_{a'} \sum_{\theta} \pi_{t'}(\theta, a) \cdot u_t(\theta, a')}_{z(t,t',a)}, 
&\forall t, t' \in T.
\label{eq:non-cred:2:substitute-z}
\end{align}
We can further substitute the maximization on the right side with an auxiliary variable $z(t,t',a)$ and force the variable to be an upper bound of the maximum value using the following additional constraints, which is linear w.r.t. $z$ and $\pi$:
\begin{align}
&z(t,t',a) \ge \sum_{\theta} \pi_{t'}(\theta, a) \cdot u_t(\theta, a'),
&\forall t, t' \in T,\, a,a' \in A.
\label{eq:non-cred:2:z-upper-bound}
\end{align}
This converts \eqref{eq:non-cred:2} to $|T|^2$ constraints in \eqref{eq:non-cred:2:substitute-z} and $|T|^2 |A|^2$ constraints in \eqref{eq:non-cred:2:z-upper-bound}, with $|T|^2 |A|$ additional variables $z(t,t',a)$.
Hence, the size of this new LP formulation is polynomial in the input size of the problem. This gives a polynomial-time algorithm to compute an optimal mechanism against a non-credible agent.

\subsection{Credible Reporting}\label{sec:cred}

Now consider a credible agent.
After reporting a type $t'$, the agent always behaves as if he is actually of type $t'$. Formally, say the agent has type $t$ and reports type $t'$. Based on the reported type $t'$ and the state $\theta$, the principal will send a signal $g \in G$ with probability $\pi_{t'}(g \given \theta)$, where $\pi_{t'}$ is the mechanism that the principal uses for a report $t'$; let $\pi_{t'}(\theta, g) = \pi_{t'}(g \given \theta) \cdot \mu(\theta)$. After receiving this signal $g$, the agent must imitate the reported type $t'$ to remain credible, so he will select an action $a'$ that satisfies
\begin{align*}
    a' \in \argmax_{a \in A} \sum_{\theta} \pi_{t'}(\theta, g) \cdot u_{t'}(\theta, a).
\end{align*}
But notice that the agent receives a payoff based on his actual payoff function $u_t$ and is equal to $\sum_{\theta} \pi_{t'}(\theta, g) \cdot u_t(\theta, a')$.

As with a non-credible agent, the principal can compute the optimal {\em IC} mechanism for a credible agent using an LP, although with a slightly different formulation (weaker constraints).
\begin{align}
\max_{(\pi_t)_{t \in T}}  \quad
& \sum_{t \in T} \rho(t) \sum_{\theta, a} \pi_t(\theta, a) \cdot v(\theta, a) \\
\text{subject to} \quad
& 
\sum_{\theta, a} \pi_t(\theta, a) \cdot u_t(\theta, a) \ge \sum_{\theta, a} \pi_{t'}(\theta, a) \cdot u_t(\theta, a) &\forall t, t' \in T \label{eq:credIC:2}\\
& 
\sum_{\theta} \pi_t(\theta, a) \cdot u_t(\theta, a) \ge \sum_{\theta} \pi_t(\theta, a) \cdot u_t(\theta, a') &\forall t \in T,\, a, a' \in A \label{eq:credIC:3} \\
&
\sum_{a} \pi_t(\theta, a) = \mu(\theta)  
&\forall t \in T,\, \theta \in \Theta \\
&
\pi_t \in \Delta(\Theta \times A) &\forall t \in T
\end{align}
Notice that constraint \eqref{eq:non-cred:2} in the LP for a non-credible agent splits into constraints \eqref{eq:credIC:2} and \eqref{eq:credIC:3} for a credible agent. Constraint \eqref{eq:non-cred:2} considers all possible deviations over reports and actions of the agent jointly, while \eqref{eq:credIC:2} and \eqref{eq:credIC:3} consider such deviations separately. Constraint \eqref{eq:credIC:2} ensures IC in the reporting stage of the game: it only considers reported-type deviations, because the principal knows that a credible agent is going to play consistently with his reported type. On the other hand, constraint \eqref{eq:credIC:3} ensures IC in the action stage of the game, by considering only action deviations; this constraint is w.l.o.g.~because we can apply the revelation principle to the last stage of the game (when the agent plays his action) and optimize only over policies that are IC (for this stage).

Although the principal can efficiently compute the optimal IC mechanism for credible agents using the LP above, the next example shows that using a non-IC mechanism can sometimes strictly improve her utility. So, the class of IC policies is {\em not} without loss of optimality for credible agents.

\begin{example}\label{ex:cred}
Let there be two agent types $T = \{t_1, t_2\}$, two states $\Theta = \{ \theta_1, \theta_2 \}$, and four actions $A  = \{ a, b, c, d\}$. The payoff functions of the agent for the two types are given in the tables below.
\begin{center}
\renewcommand{\arraystretch}{1.3}
\small
    \begin{tabular}{L | R | R | R | R |}
        \myheader{} & \myheader{a} & \myheader{b} &  \myheader{c} &  \myheader{d} \\
         \cline{2-5}
         \theta_1 & 4 & 1 & 3 & 0 \\
         \cline{2-5}
         \theta_2 & 1 & 4 & 0 & 3 \\
         \cline{2-5}
        \mytablecaption{} & \multicolumn{4}{c}{Type $t_1$} 
    \end{tabular}
    \qquad
    \begin{tabular}{ L | R | R | R | R |}
        \myheader{} & \myheader{a} & \myheader{b} & \myheader{c} & \myheader{d} \\
         \cline{2-5}
         \theta_1 & 0 & 0 & 1 & 0 \\
         \cline{2-5}
         \theta_2 & 0 & 0 & 0 & 1 \\
         \cline{2-5}
         \mytablecaption{} & \multicolumn{4}{c}{Type $t_2$}
    \end{tabular}
\end{center}
The principal gets payoff $1$ if the agent plays actions $c$ or $d$ and payoff $0$ if the agent plays actions $a$ or $b$, irrespective of the state. Let the prior be uniform, $\mu(\theta_1) = \mu(\theta_2) = 1/2$.
\end{example}

Let us first consider what happens when a credible agent reports each of the possible types in the above example.
\begin{itemize}
\item
If the agent reports type $t_2$, then by credibility she must act like an agent of type $t_2$, irrespective of her actual type. Notice that for type $t_2$, the actions in $\{c, d\}$ dominate those in $\{a,b\}$. In particular, for any posterior belief about the state, say $(p, 1-p)$, that the agent may have after receiving principal's signal, he has payoff $0$ for playing $a$ or $b$ but payoff $\max(p, 1-p) > 0$ for playing the better option among $c$ or $d$. So, after reporting type $t_2$, the agent always plays an action that gives the principal payoff $1$.

\item
On the other hand, if the agent reports type $t_1$, she must always play an action in $\{a, b\}$ and never an action in $\{c,d\}$. In particular, action $a$ dominates action $c$, and action $b$ dominates $d$. As the agent never plays $c$ or $d$, the principal only gets payoff $0$.
\end{itemize}

In an IC mechanism, an agent reports type $t_i$ only when he is actually $t_i$, which happens with probability $\rho(t_i)$. Therefore, the expected payoff of the principal is $0 \times \rho(t_1) + 1 \times \rho(t_2) = \rho(t_2)$. On the other hand, we propose the following non-IC mechanism that achieves an expected payoff $1$, which is strictly greater than $\rho(t_2)$ if $\rho(t_1) > 0$.

The non-IC mechanism works as follows: If the agent reports $t_1$, the principal does not give any information about the state to the agent, but if the agent reports $t_2$, the principal fully reveals the state. Intuitively, the principal trades information with the agent's ``promise'' to act as the more preferred type $t_2$.
This is effective only when the agent is credible.
Consider the response of a type-$t_1$ agent under this mechanism.
\begin{itemize}
\item 
If he reports $t_1$, as the principal does not reveal any information, the agent's posterior belief is the same as the prior $(1/2, 1/2)$, and playing either of $a$ or $b$ gives him an expected payoff of $5/2$. 
\item
If he reports $t_2$, the principal reveals the state, and depending upon whether the state is $\theta_1$ or $\theta_2$, he plays either $c$ or $d$ (because he must imitate $t_2$) and gets a payoff of $3$. 
\end{itemize}
So, the type-$t_1$ agent is incentivized to report and act like $t_2$. It is trivial to check that an agent of type $t_2$ will report $t_2$, too. Consequently, all agent types will report $t_2$ and play either $c$ or $d$, and the principal gets expected payoff $1$. The ratio of between the expected payoffs for the principal for the optimal and optimal-IC mechanism is $1/\rho(t_2) \rightarrow \infty$ as $\rho(t_2) \rightarrow 0$.

\smallskip

Given that IC mechanisms may not be optimal, the principal may want to optimize over all policies, including non-IC mechanisms, too. The theorem below says that even to compute an approximately optimal mechanism is NP-hard. The proof uses a tight reduction from {\sc Maximum Independent Set}, similarly to the approach by \citet{gan2019imitative}.\footnote{All omitted proofs can be found in the appendix.}

\begin{restatable}{theorem}{thmnphnoIC}
\label{thm:nph-no-IC}
Unless P=NP, there exists no polynomial-time $\frac{1}{(|T|-1)^{1-\epsilon}}$-approximation algorithm for computing a mechanism (against a credible agent) for any constant $\epsilon > 0$, even when there are only two possible states of nature.
\end{restatable}

In the remainder of the paper, unless otherwise specified all our results assume a credible agent.

\section{Moving Beyond Two Stages}\label{sec:multi_stage}

We have seen that the principal may benefit by using non-IC mechanisms. As it turns out, when credibility is required, there are even more instruments the principal can leverage to further improve her payoff. 
In this section, we demonstrate two such instruments:
(1) pre-signaling, and 
(2) non-binding elicitation. 
Pre-signaling means that the principal sends a signal before she elicits the agent's type. Non-binding elicitation is a cheap-talk communication in addition to the agent's type reporting.
It asks the agent how they would like the mechanism to proceed. Since the question is not type related, it is non-binding; no restrictions are attached to the agent's answer on how they must act subsequently. 

These additional approaches naturally result in mechanisms involving more stages of interaction between the principal and the agent, beyond the two-stage elicitation-then-signaling paradigm studied in the previous section; see \Cref{fig:timeline-presign-nonbind}. 
We demonstrate the strict payoff improvement of these mechanisms, their general optimality, and more surprisingly, the computational benefits they bring: in sharp contrast to the intractability of the two-stage mechanisms, we present a polynomial time algorithm to optimize this new type of mechanism.

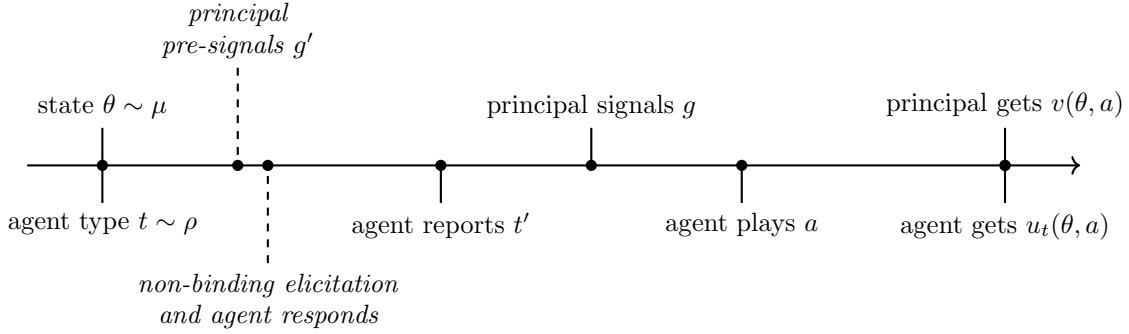
\begin{figure}
    \centering
    \small
    \begin{tikzpicture}
      \draw[->,thick] (-1,0) -- (13, 0);
      \draw[-,thick] (0,-0.5) node[below, align=center]{agent type $t \sim \rho$} -- (0,0) node[circle,fill,inner sep=1.5pt]{};
      \draw[-,thick] (0, 0.5) node[above, align=center]{state $\theta \sim \mu$} -- (0,0);
      \draw[-,thick] (4.5,-0.5) node[below, align=center]{agent reports $t'$} -- (4.5,0) node[circle,fill,inner sep=1.5pt]{};
      \draw[-,thick] (6.5, 0.5) node[above, align=center]{principal signals $g$} -- (6.5,0) node[circle,fill,inner sep=1.5pt]{};
      \draw[-,thick] (8.5,-0.5) node[below, align=center]{agent plays $a$} -- (8.5,0) node[circle,fill,inner sep=1.5pt]{};
      \draw[-,thick] (12,0.5) node[above, align=center]{principal gets $v(\theta,a)$} -- (12,0) node[circle,fill,inner sep=1.5pt]{};
      \draw[-,thick] (12,-0.5) node[below, align=center]{agent gets $u_t(\theta,a)$} -- (12,0) node[circle,fill,inner sep=1.5pt]{};
      \draw[-,dashed,thick] (1.8, 1.3) node[above, align=center]{\em principal\\ \em pre-signals $g'$} -- (1.8,0) node[circle,fill,inner sep=1.5pt]{};
      \draw[-,dashed,thick] (2.2, -1.3) node[below, align=center]{\em non-binding elicitation \\ \em and agent responds} -- (2.2,0) node[circle,fill,inner sep=1.5pt]{};
    \end{tikzpicture}
    \caption{Timeline with pre-signaling and non-binding elicitation.}
    \label{fig:timeline-presign-nonbind}
\end{figure}

\subsection{Pre-signaling}

Intuitively, pre-signaling aims to change the agent's belief to influence their decision on which type to imitate. In the following example, we compare mechanisms with/without pre-signaling, and show that pre-signaling offers a strict improvement in payoff for the principal. 
For ease of description, we refer to a mechanism that does not use pre-signaling as an $\ES$ mechanism, where $\xE$ and $\xS$ represents the elicitation and signaling stages, respectively, in the mechanism.
Similarly, we refer to a mechanism that uses pre-signaling as an $\SES$ mechanism, where the first $\xS$ refers to the additional pre-signaling stage.

\begin{example}
\label{exp:SESA-gt-ESA}
Let $\Theta = \{\alpha, \beta, \gamma \}$, $T = \{t_0, t_1\}$, and $A = \{a, a', b, b'\}$.
The principal's payoff is $1$ for actions $a'$ and $b'$, and $0$ for actions $a$ and $b$, irrespective of the state. The agent's payoffs are given in the table below. Moreover, the type and state distributions are uniform: $\rho(t_0) = \rho(t_1) = 1/2$, and  $\mu(\gamma) = \mu(\alpha) = \mu(\beta) = 1/3$.
\begin{center}
\renewcommand{\arraystretch}{1.3}
\small
    \begin{tabular}{ L | R | R | R | R |}
         \myheader{} & \myheader{a} & \myheader{a'} & \myheader{b} & \myheader{b'} \\
         \cline{2-5}
         \alpha,\beta,\gamma & 0 & 1 & 0 & 1 \\
         \cline{2-5}
         \mytablecaption{} &\multicolumn{4}{c}{Type $t_0$}
    \end{tabular}
    \qquad
    \begin{tabular}{ L | R | R | R | R |}
         \myheader{} & \myheader{a} & \myheader{a'} & \myheader{b} & \myheader{b'} \\
         \cline{2-5}
         \alpha & 3 & 2 & -2 & -3  \\
         \cline{2-5}
         \beta & -2 & -3 & 3 & 2  \\
         \cline{2-5}
         \gamma & 0 & -10 & 0 & -10 \\
         \cline{2-5}
         \mytablecaption{} &\multicolumn{4}{c}{Type $t_1$}
    \end{tabular}
\end{center}
\end{example}

In this example, the principal's incentive is completely aligned with a type-$t_0$ agent but in conflict with type $t_1$. For $t_1$, actions $a$ and $b$ strictly dominate $a'$ and $b'$, respectively. Hence, it is preferred that the agent acts according to $t_0$. 
We will demonstrate that no $\ES$ mechanism can persuade a type-$t_1$ agent to imitate $t_0$, while there exists an $\SES$ mechanism that achieves this with a positive probability.

\subsubsection{Without Pre-signaling}
\label{sc:with-pre-signaling}

When there is no pre-signaling, the key observation (\Cref{clm:ESA-t1-no-misreport}) is that a type-$t_1$ agent will always report truthfully. 
Since actions $a'$ and $b'$ are strictly dominated by $a$ and $b$, a type-$t_1$ agent will then never choose $a'$ or $b'$. As a result, the principal gets payoff $0$ when the agent's true type is $t_1$, which happens with probability $\rho(t_1) = 0.5$. Overall, the principal's expected payoff is at most $0.5 \times 0 + 0.5 \times 1 = 0.5$ (where $1$ is the principal's maximum attainable payoff according to the example).

\begin{lemma}
\label{clm:ESA-t1-no-misreport}
In \Cref{exp:SESA-gt-ESA}, a type-$t_1$ agent has no incentive to misreport his type as $t_0$ under any $\ES$ mechanism.
\end{lemma}

\begin{proof}
Consider an arbitrary $\ES$ mechanism.
For a type-$t_1$ agent, even without any additional signal from the principal, they are able to secure an expected payoff of $1/3$ by sticking to action $a$. 
Now suppose they misreport type $t_0$, and we show that this leads to a contradiction.

For a type-$t_0$ agent, actions $a'$ and $b'$ strictly dominate $a$ and $b$, respectively.
Hence, when an agent acts as $t_0$, they will play only $a'$ or $b'$. This means the following for a type-$t_1$ agent for every $\theta \in \Theta$:
\[
\Pro(a' \given \theta, t_1) + \Pro(b' \given \theta, t_1)  = 1,
\]
where $\Pro$ denotes the probability measure induced by the $\ES$ mechanism (and the agent's best response to this mechanism).
It follows that the payoff of the type-$t_1$ agent would be:
\begin{align}
& \sum_{\theta \in \Theta} \sum_{x' \in \{a',b'\}}  \Pro(x', \theta \given t_1) \cdot u_{t_1}(\theta, x') 
\ =\ 
\sum_{\theta \in \Theta} \mu(\theta) \cdot \Big ( \Pro(a' \given \theta, t_1) \cdot u_{t_1}(\theta, a') + \Pro(b' \given \theta, t_1) \cdot u_{t_1}(\theta, b') \Big) \nonumber\\
&\qquad 
\le\
\sum_{\theta \in \Theta} \mu(\theta) \cdot \max \big\{ u_{t_1}(\theta, a'),\, u_{t_1}(\theta, b') \big\} 
\ =\ 
1/3 \cdot 2 + 1/3 \cdot 2 + 1/3 \cdot (-10) 
\ <\ 
0, 
\label{eq:type-t1-ub}
\end{align}
which is strictly lower than the minimum payoff $1/3$ that is guaranteed by reporting truthfully. Therefore, a type-$t_1$ agent is always better off reporting truthfully.
\end{proof}

\subsubsection{With Pre-signaling}
\label{sc:without-pre-signaling}

Now, consider the following mechanism that uses pre-signaling.
Suppose that, before the agent reports their type, the principal pre-signals $g_\gamma$ when the state is $\gamma$, and pre-signals $g_{\neg\gamma}$ when it is any other state.
Depending on the pre-signal sent, the game proceeds as follows:
\begin{enumerate}
\item When $g_\gamma$ is sent, the agent knows with certainty that the state is $\gamma$.
\item When $g_{\neg\gamma}$ is sent, the agent's belief is a uniform distribution over $\alpha$ and $\beta$.
\end{enumerate}
In case~2, by using the following signaling strategy, the principal can obtain payoff $1$ subsequently:
\begin{align*}
\pi_{t_1} ( a \mid \alpha) = \pi_{t_1} ( a \mid \beta) =1, 
\qquad\text{and }\qquad
\pi_{t_0} ( a' \mid \alpha) = \pi_{t_0} ( b' \mid \beta) =1.
\end{align*}
Specifically, $\pi_{t_1}$ is uninformative as it always recommends $a$, whereas $\pi_{t_0}$ completely differentiates $\alpha$ and $\beta$. 
Under this mechanism, agents of both types will be incentivized to always (mis)report $t_0$ and perform the recommended actions. The principal obtains payoff $1$ as a result.
Overall, since the marginal probability of case~2 is $\mu(\alpha) + \mu(\beta) = 2/3$, the principal can obtain payoff at least $2/3$ in expectation. This is a strict improvement compared with $0.5$ in the case without pre-signaling.

\medskip

The analysis in \Cref{sc:with-pre-signaling,sc:without-pre-signaling} then implies the following.
\begin{theorem}
There exists an instance where an optimal $\SES$ mechanism yields a strictly higher payoff for the principal than any $\ES$ mechanism does.
\end{theorem}

\subsection{Non-binding Elicitation}

Non-binding elicitation can be thought of as a menu for the agent to choose how they would like the mechanism to proceed.
An example of mechanisms with non-binding elicitation is illustrated in \Cref{fig:non-binding} (we will shortly demonstrate that it offers strict payoff improvement in \Cref{exp:EnSESA-gt-SESA}).
In this mechanism, a non-binding elicitation stage is arranged at the beginning of the mechanism. We will use the notation $\xEN$ to represent non-binding elicitation.
At $\xEN$, the principal asks the agent to choose between two options labelled 1 and 2, respectively. Depending on the agent's choice, the principal will subsequently proceed with one of the $\SES$ mechanisms. The elicitation at $\xEN$ does not involve any type information and is therefore \emph{non-binding} by nature: the agent's report at $\xEN$ does not result in any credibility constraints on how they should behave subsequently.

\tikzset{
E node/.style={rounded corners, draw, text width=3mm, inner sep=5pt, align=center, thick}
}

\tikzset{
S node/.style={circle, draw, text width=5mm, inner sep=2pt, align=center, thick}
}

Following our previously established notation, we will refer to such a mechanism as an $\EnSES$ mechanism.
We demonstrate in the following example that there exists an $\EnSES$ mechanism that strictly outperforms any $\SES$ mechanism.
\begin{example}
\label{exp:EnSESA-gt-SESA}
Let $\Theta = \{\alpha, \beta, \gamma\}$, $T = \{t_1, t_2\}$, and $A = \{a, a', b, b', d\}$.
Irrespective of the state, the principal's payoff is $1$ for actions $a'$ and $b'$, and $0$ for all the other actions.
The agent's payoffs are given in the tables below, where all blank entries are $-M$, for some sufficiently large constant, say $M > 10000$.  
The type and state distributions are uniform: $\mu(\alpha) = \mu(\beta) = \mu(\gamma) = 1/3$ and $\rho(t_1) = \rho(t_2) = 1/2$.
\begin{center}
\renewcommand{\arraystretch}{1.3}
\small
\begin{tabular}{L | R | R | R | R | R |}
         \myheader{} & \myheader{a} & \myheader{a'} & \myheader{b} & \myheader{b'} & \myheader{d} \\
         \cline{2-6}
         \alpha &  & 1 &  & & 0 \\
         \cline{2-6}
         \beta &  &  & 3 & 2 & 0 \\
         \cline{2-6}
         \gamma &  & -1 &  &  & 0 \\
         \cline{2-6}
         \mytablecaption{} & \multicolumn{5}{c}{Type $t_1$} 
    \end{tabular}
    \qquad
    \begin{tabular}{L | R | R | R | R | R |}
         \myheader{} & \myheader{a} & \myheader{a'} & \myheader{b} & \myheader{b'} & \myheader{d} \\
         \cline{2-6}
         \alpha & 3 & 2 &  & & 0 \\
         \cline{2-6}
         \beta &  &  &  & 1 & 0 \\
         \cline{2-6}
         \gamma &  &  &  & -1 & 0 \\
         \cline{2-6}
         \mytablecaption{} & \multicolumn{5}{c}{Type $t_2$} 
    \end{tabular}
\end{center}
\end{example}

\tikzset{
  arrleaf/.style={-{Circle[black,length=4pt]}}
}

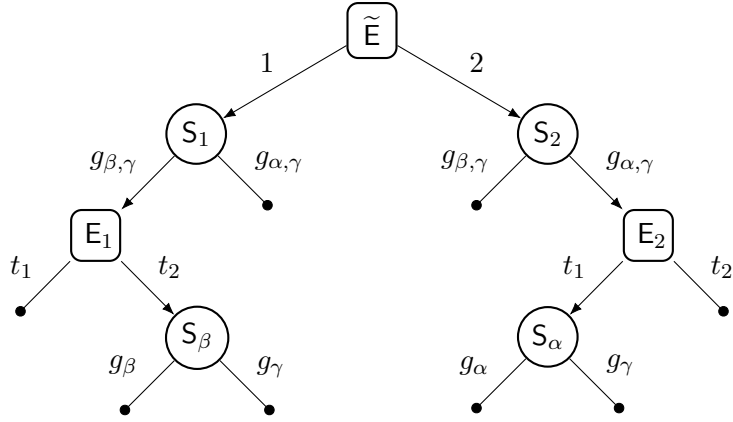
\begin{figure*}
\centering
\begin{tikzpicture}
    \node[E node] (1) {$\xEN$};
    \node[S node] (2) [below left = 1cm of 1, xshift=-10mm] {$\xS_1$};
    \node[S node] (3) [below right = 1cm of 1, xshift=10mm] {$\xS_2$};
    \node[E node] (4) [below left = 1cm of 2] {$\xE_1$};
    \node (5) [below right = 1cm of 2] {~};
    \node (6) [below left = 1cm of 3] {~};
    \node[E node] (7) [below right = 1cm of 3] {$\xE_2$};
    \node (8) [below left = 1cm of 4] {~};
    \node[S node] (9) [below right = 1cm of 4] {$\xS_\beta$};
    \node (10) [below right = 1cm of 7] {~};
    \node[S node] (11) [below left = 1cm of 7] {$\xS_\alpha$};
    \node (12) [below right = 1cm of 9] {~};
    \node (13) [below left = 1cm of 9] {~}; 
    \node (14) [below right = 1cm of 11] {~};
    \node (15) [below left = 1cm of 11] {~}; 
    \draw (1) -- node [black, midway, above left] {$1$} (2) [anchor=west];
    \draw (1) -- node [black, midway] {$2$} (3) [anchor=west];
    \draw (2) -- node [black, midway, above left] {$g_{\beta,\gamma}$} (4) [anchor=west];
    \draw[arrleaf] (2) -- node [black, midway, above right] {$g_{\alpha,\gamma}$} (5) [anchor=west];
    \draw[arrleaf] (4) -- node [black, midway, above left] {$t_1$} (8) [anchor=west];
    \draw (4) -- node [black, midway] {$t_2$} (9) [anchor=west];
    \draw[arrleaf] (3) -- node [black, midway, above left] {$g_{\beta,\gamma}$} (6) [anchor=west];
    \draw (3) -- node [black, midway] {$g_{\alpha,\gamma}$} (7) [anchor=west];
    \draw[arrleaf] (7) -- node [black, midway] {$t_2$} (10) [anchor=west];
    \draw (7) -- node [black, midway, above left] {$t_1$} (11) [anchor=west];
    \draw[arrleaf] (9) -- node [black, midway] {$g_\gamma$} (12) [anchor=west];
    \draw[arrleaf] (9) -- node [black, midway, above left] {$g_\beta$} (13) [anchor=west];
    \draw[arrleaf] (11) -- node [black, midway] {$g_\gamma$} (14) [anchor=west];
    \draw[arrleaf] (11) -- node [black, midway, above left] {$g_\alpha$} (15) [anchor=west];
\end{tikzpicture}
\caption{A mechanism that uses non-binding elicitation.}
\label{fig:non-binding}
\end{figure*}

\subsubsection{With Non-binding (and Binding) Elicitation}

We first present an $\EnSES$ mechanism that yields a payoff strictly greater than $1 - \delta$ for the principal in \Cref{exp:EnSESA-gt-SESA}, where $\delta = 10/M$.  
The mechanism is illustrated in \Cref{fig:non-binding}. The non-binding elicitation at the beginning of the mechanism offers two options $1$ and $2$.  
The signaling strategies used at the signaling nodes $\xS_1$, $\xS_3$, $\xS_\alpha$, and $\xS_\beta$ are defined as follows.
\begin{itemize}
\item
We use $\pi_1$ at $\xS_1$, where:
\begin{align*}
\pi_1(g_{\alpha,\gamma} \given \alpha) = 1, \qquad
\pi_1(g_{\beta,\gamma} \given \beta) = 1,
\qquad\text{and}\quad
\begin{cases}
\pi_1(g_{\alpha,\gamma} \given \gamma) = 1 - \delta \\
\pi_1(g_{\beta,\gamma} \given \gamma) = 
\delta
\end{cases}.
\end{align*}
Namely, $\pi_1$ is designed in a way that induces the following posterior beliefs about the state when the agent receives $g_{\alpha,\gamma}$ and $g_{\beta,\gamma}$:
\begin{talign}
\label{eq:EnSESA-posteriors}
\Pro( \cdot \given g_{\alpha,\gamma}, \xS_1) = \left(\ \frac{1}{2 - \delta},\   0,\   \frac{1 - \delta}{2 - \delta}\  \right)
\quad \text{and} \quad
\Pro( \cdot \given g_{\beta,\gamma}, \xS_1) = \left(\  0,\   \frac{1}{1 + \delta},\quad  \frac{\delta}{1 + \delta}\  \right),
\end{talign}
where $\Pro$ denotes the probability measure induced by the mechanism and the agent's best response.
The states in the subscript of each signal correspond to states in the support of the posteriors.
It will also be useful to note that the marginal probability of each signal is:
\begin{align*}
\Pro(g_{\alpha, \gamma} \given \xS_1) 
= (2-\delta)/3,
\quad\text{and}\quad
\Pro(g_{\beta, \gamma} \given \xS_1)  = (1+\delta)/3.
\end{align*}

\item
The strategy $\pi_2$ used at $\xS_2$ is defined symmetrically: we swap $\alpha$ with $\beta$ in the definition of $\pi_1$ to obtain $\pi_2$.

\item 
The strategies $\pi_\alpha$ and $\pi_\beta$ used at $\xS_\alpha$ and $\xS_\beta$ are truth-revealing strategies: $\pi_\alpha(g_\theta \given \theta) = \pi_\beta(g_\theta \given \theta) = 1$ for all $\theta \in \{\alpha, \beta, \gamma\}$.
\end{itemize}
At each leaf node, the mechanism terminates and the agent takes an action that complies with the credibility constraint.
If no type elicitation has taken place along the path from the root to the leaf node, the agent is free to take an optimal action of any type---without loss of generality, they will just take an optimal action according to their actual type.

\smallskip

The mechanism ensures the principal a payoff of $1-\delta$, which is close to her payoff upper bound $1$ in this example.
Indeed, to ensure this high payoff, the agent needs to be persuaded to play $a'$ or $b'$ almost all the time. 
Examining each state separately, we will find that the incentives of types $t_1$ and $t_2$ align with the this goal at states $\alpha$ and $\beta$, respectively, but not at the other states.
Hence, we need to incentivize a type-$t_1$ agent to act as $t_2$ with a sufficiently high probability at state $\beta$, and incentivize a type-$t_2$ agent to act as $t_1$ at state $\alpha$.
This is what the mechanism does: it keeps $\alpha$ and $\beta$ mixed with $\gamma$ (at $\xS_1$ and $\xS_2$, respectively); subsequently, it ``disentangles'' $\alpha$ and $\beta$ from $\gamma$ (at $\xS_\alpha$ and $\xS_\beta$, respectively) only if the agent promises to act as $t_2$ and $t_1$, respectively.

The non-binding elicitation plays a role here because the ways we mix and disentangle the states need to be designed differently for the two agent types. Indeed, the following lemma shows that the design incentivizes each type $t_i$ to choose the corresponding option~$i$ at $\xEN$.   

\begin{restatable}{lemma}{EnSESAoptioni}
\label{lmm:EnSESA-gt-SESA-option-i}
In \Cref{exp:EnSESA-gt-SESA}, given the mechanism presented in \Cref{fig:non-binding}, for each $i \in \{1,2\}$ the optimal strategy of a type-$t_i$ agent is to select option~$i$ at $\xEN$.
\end{restatable}

Following the lemma, by symmetry, the principal obtains the same payoff on $t_1$ and $t_2$. 
Consider the situation against a type-$t_1$ agent. It can be verified that
(see the proof of \Cref{lmm:EnSESA-gt-SESA-option-i} for more details): 
\begin{itemize}
\item 
When $g_{\alpha,\gamma}$ is sent, the agent will react optimally by playing action $a'$, and the principal obtains payoff $1$.

\item 
When $g_{\beta,\gamma}$ is sent, the agent acts as $t_2$, playing $b'$ when it is $\beta$ and $d$ when it is $\gamma$.
The principal obtains payoff $\Pro(\beta \given g_{\beta,\gamma}, \xS_1) \cdot 1 = 1/(1+\delta)$.  

\end{itemize}
Overall, the principal's payoff is 
\begin{align*}
\Pro(g_{\alpha,\gamma} \given \xS_1) \times 1 
+ 
\Pro(g_{\beta,\gamma} \given \xS_2) \times \frac{1}{1 + \delta} 
\ =\
\frac{2 - \delta}{3} \times 1 
+
\frac{1 + \delta}{3} \times \frac{1}{1 + \delta} 
\ >\ 
1 - \delta.
\end{align*}

\subsubsection{With Binding Elicitation Only}

Consider an arbitrary $\SES$ mechanism $\Pi$,\footnote{We will hereafter mostly use $\Pi$ to denote a multi-stage mechanism and use $\pi$ to denote the signaling strategy in $\Pi$ for the pre-signaling stage.} which only uses binding elicitation.
We will argue that $\Pi$ yields payoff at most $1-\delta$ for the principal.
In what follows, let $\pi: \Theta \to \Delta(G)$ be the signaling strategy used in the pre-signaling step of $\Pi$, such that $\pi(g) > 0$ for all $g \in G$.
Let $\Pro$ denote the probability measure induced by $\Pi$ (and the agent's optimal response) and let $\mathbb{E}$ denote the expectation over $\Pro$.

The proof of the $1-\delta$ upper bound relies on the following lemmas. \Cref{lmm:u-g} first argues that if $\Pi$ were to achieve a high payoff for the principal, there must be a high-payoff signal $g$ that implies $\gamma$ with a sufficiently high probability. (Recall that $v$ denotes the principal's payoff function.) 
\Cref{lmm:g-alpha} further shows that the principal must receive a high-payoff signal that  implies both $\alpha$ and $\beta$ with a sufficiently high probability.
With these observations, we prove the $1- \delta$ upper bound on the performance of all $\SES$ mechanisms in \Cref{exp:EnSESA-gt-SESA} and establish \Cref{thm:exp:EnSESA-gt-SESA}.

\begin{restatable}{lemma}{lmmug}
\label{lmm:u-g} 
Suppose that $\mathbb{E}(v) \ge 1 - \delta$ in \Cref{exp:EnSESA-gt-SESA}.
Then, for any $\epsilon \in (0,\, 1 - \delta)$, there exists a signal $g \in G$ such that $\mathbb{E}(v \given g) > 1 - \delta - \epsilon$ and $\Pro(\gamma \given g) \ge 1/4 - \delta/\epsilon$.
\end{restatable}

\begin{proof}[Proof sketch]
The converse of the statement asserts the existence of $\epsilon$ such that: for every $g \in G$, if $\Pro(\gamma \given g) \ge 1/4 - \delta/\epsilon$ then $\mathbb{E}(v \given g) \le 1 - \delta - \epsilon$. 
Indeed, since the prior probability is $\mu(\gamma) = 1/4$, the marginal probability of signals $g$ such that $\Pro(\gamma \given g) \ge 1/4 - \delta/\epsilon$ cannot be insignificant.
Hence, if the converse statement is true, then for a substantial portion of signals would yield expected payoff lower than $1 - \delta - \epsilon$ for the principal, which would further lead to an overall payoff strictly lower than the assumed value $1 - \delta$.   
\end{proof}

\begin{restatable}{lemma}{lmmgalpha} 
\label{lmm:g-alpha}
Suppose that $\mathbb{E}(v) \ge 1 - \delta$ in \Cref{exp:EnSESA-gt-SESA}.
Then there exists a signal $g \in G$ such that $\mathbb{E}(v \given g) > 1 - 11\delta$,
$\Pro(\alpha \given g) > 1/100$, and $\Pro(\beta \given g) > 1/100$.
\end{restatable}

\begin{proof}[Proof sketch]
By \Cref{lmm:u-g} some high-payoff $g$ must imply $\gamma$ with a significant probability. 
Now that there is no non-binding elicitation, both agent types will receive $g$ with some probability. For type $t_1$ (and similarly for type $t_2$ and state $\beta$), the probability of $\alpha$ in the posterior conditioned on $g$ must be sufficiently high: otherwise either $b$ or $d$ will always strictly dominate other actions (e.g., consider the extreme case where $\Pro(\alpha\given g) = 0$), leading to a low payoff $0$ for the principal and a contradiction to the assumption that $\mathbb{E}(v \given g) > 1 - 11\delta$.
\end{proof}

\begin{theorem}
\label{thm:exp:EnSESA-gt-SESA}
There exists an instance where an optimal $\EnSES$ mechanism yields a strictly higher payoff for the principal than any $\SES$ mechanism.
\end{theorem}

\begin{proof}
We show that in \Cref{exp:EnSESA-gt-SESA} no $\SES$ mechanism yields payoff higher than $1-\delta$ for the principal.
Consider an arbitrary $\SES$ mechanism and suppose for the sake of contradiction that it yields payoff $\mathbb{E}(v) > 1-\delta$ for the principal.

Note that the mechanism elicits the agent's type in the second step.
Consider the case where the agent is of type $t_1$ and a signal $g$ that satisfies the conditions in \Cref{lmm:g-alpha} is sent. Subsequently, when the mechanism elicits the agent's type, the agent either reports $t_1$ or $t_2$.
Consider each of these two possible reports. We show that both cases lead to contradictions to complete the proof.

\paragraph{Case 1: $t_1$ is reported.}

In this case, the type-$t_1$ agent will act truthfully. 
Consider the agent's expected payoff $\mathbb{E} [ u_{t_1} \given t_1, g]$ and the following upper bound on the payoff:
\begin{align*}
\mathbb{E} [ u_{t_1} \given t_1, g] 
&\le
\Pro(\beta, a \vee a' \given t_1, g) \cdot (-M) + 
(1 - \Pro(\beta, a \vee a' \given t_1, g)) \cdot 3.
\end{align*}
It must be $\mathbb{E} [ u_{t_1} \given t_1, g] 
\ge 0$ because the agent can stick to action $d$ to secure payoff $0$.
This gives 
$\Pro(\beta, a \vee a' \given t_1, g) \le \frac{3}{M + 3}$.
Note that, for type $t_1$, action $b'$ is strictly dominated by $b$ or $d$ under any state distribution, which means 
$\Pro(\beta, b' \given t_1, g) = 0$. 
Therefore, 
\begin{talign*}
\Pro(\beta, b \vee d \given t_1, g) = \Pro(\beta \given t_1, g) - \Pro(\beta, a \vee a' \given t_1, g)
> \frac{1}{100} - \frac{3}{M + 3}.
\end{talign*}
Since the principal gets payoff $0$ when $b$ or $d$ is played,
we get that
\begin{talign*}
\mathbb{E}(v \given g) 
&
\le 1 -  \Pro(b \vee d \given g) 
\le 1 -  \Pro(t_1, \beta, b\vee d \given g) \\
&\qquad\qquad\qquad
= 1 - \Pro(t_1 \given g) \cdot \Pro(\beta, b \vee d \given t_1, g)
\le
1 - \frac{1}{3} \cdot \left( \frac{1}{100} - \frac{3}{M+3} \right) 
< 
1 - 11 \delta
\end{talign*}
(note that $\Pro(t_1 \given g) = \Pro(t_1) = \rho(t_1) = 1/3$ by independence).
This contradicts \Cref{lmm:g-alpha}.

\paragraph{Case 2: $t_2$ is reported.}

In this case, the agent will act according to type $t_2$.
We can replicate the above arguments to show that 
$\Pro(\alpha, a \vee d \given t_1, g)
> \frac{1}{100} - \frac{3}{M + 3}$,
where we replace $\beta$ with $\alpha$, $a$ and $a'$ with $b$ and $b'$, respectively, and $u_{t_1}$ with $u_{t_2}$.
Eventually, this also leads to the contradiction $\mathbb{E}(v \given g) < 1 - 11 \delta$.
\end{proof}

\section{Computing Optimal Multi-stage Mechanisms}
\label{sec:compute}

$\EnSES$ and $\SES$ mechanisms are strictly better than $\ES$ mechanisms in terms of the payoffs they generate. But what do they imply computationally? Is it easier or harder to compute optimal $\EnSES$ and $\SES$ mechanisms?
As it turns out, optimal $\SES$ mechanisms are still inapproximable by noting that in the reduction for \Cref{thm:nph-no-IC} the principal does not benefit from using pre-signaling (\Cref{thm:nph-SESA}). 
But optimal $\EnSES$ mechanisms, while generating even higher payoffs than $\SES$s, are efficiently computable. 
We next present an efficient algorithm for $\EnSES$ and demonstrate its general optimality among a broad class of multi-stage mechanisms.

\begin{theorem}
\label{thm:nph-SESA}
Unless P = NP, there exists no polynomial-time $\frac{1}{(|T|-1)^{1-\epsilon}}$-approximation algorithm for computing an $\SES$ mechanism for any constant $\epsilon > 0$, even when there are only two possible states of nature.
\end{theorem}

\begin{proof}
We use the same reduction as that in the proof of \Cref{thm:nph-no-IC} and note that the principal does not benefit from using pre-signaling in the reduced instance. 
\end{proof}

\subsection{The General Optimality of $\EnSES$}

The efficient algorithm relies on a characterization of optimal $\EnSES$ mechanisms.
The characterization is quite general: it compares $\EnSES$ mechanisms with all possible multi-stage mechanisms involving any number of signaling and binding or non-binding elicitation steps. We hereafter refer to this broader class of mechanisms as {\em indefinite-stage mechanisms}. 
Similarly to the mechanisms we presented previously, an indefinite-stage mechanism can be history-dependent, specifying different strategies for different nodes of the game tree. For binding elicitation, we consider only direct elicitation, where the principal asks the agent directly ``what is your type?'' and the agent must choose one type in $T$ as their answer to this question. 
(In the next section, we will further investigate a more general type of elicitation that allows the principal to elicit partial type information.)

For ease of description, in what follows, we assume that the type set contains $n$ types: $T = \{t_1, \dots, t_n\}$. Since an $\EnSES$ mechanism corresponds to a tree as we illustrated in \Cref{fig:non-binding}, we will describe the mechanism based on the tree. The tree contains the following types of nodes:
\begin{itemize}
\item 
$\xS$ nodes. An $\xS$ node represents a signaling step, where the principal sends a signal according to a strategy $\pi: \Theta \to \Delta(G)$ for this node. The node then has $|G|$ child nodes each corresponding to a distinct signal in $G$. When a signal $g$ is drawn, the mechanism proceeds to the child corresponding to $g$.

\item 
$\xE$ nodes. An $\xE$ node represents to a binding elicitation step, where the agent reports their type. Each $\xE$ node has $n$ children, each corresponding to a type in $T$. When the agent reports a type $t$, the mechanism proceeds to the child corresponding to $t$.

\item
$\xEN$ nodes. An $\xEN$ node represents to a non-binding elicitation step, where the agent selects a child of this node, and the mechanism proceeds to the selected child.
The root node of an $\EnSES$ mechanism is always an $\xEN$ node.

\item 
Leaf nodes.
At a leaf node, the mechanism terminates, and the agent performs an optimal action of the {\em reported type}. If no type elicitation takes place on the path to the leaf, we assume that the agent is free to choose an optimal action of any type.\footnote{W.l.o.g., the agent will choose an optimal action of their actual type. In fact, it is w.l.o.g. to assume that there is always one $\xE$ node on every path from the root to a leaf (see the proof of \Cref{lmm:EnSEA-optimal}).} 
\end{itemize}
To sidestep intricate corner cases, we follow the convention and assume that if there is a tie in the agent's decision-making (whether to decide an action to play or a type to report), the mechanism decides how to break the tie.

The characterization result, presented in the theorem below can be viewed as a generalized revelation principle that involves multiple levels of IC or DIC in an $\EnSES$ mechanism (where ``D'' emphasizes {\em direct} recommendation). Intuitively, the option the agent selects at $\xEN$ reveals their true type, though this does not mean that the agent needs to subsequently act truthfully because of the non-binding nature of this elicitation.  
Subsequently, the signaling strategy at each $S_i$ directly recommends the agent a type to report, while those at $S_{ij}$ directly recommend actions.

\begin{restatable}{theorem}{lmmEnSEAoptimal}
\label{lmm:EnSEA-optimal}
For any indefinite-stage mechanism $\Pi$, there exists an $\EnSES$ mechanism $\Pi'$ that yields as much payoff for the principal as $\Pi$ does and has the following properties.
\begin{itemize}
\item[i.] {\bf IC at $\xEN$}:
The root $\xEN$ has $n$ children $\xS_1, \dots, \xS_n$, which are all $\xS$ nodes.
Each type-$t_i$ agent is incentivized to choose the branch leading to $\xS_i$.

\item[ii.] {\bf DIC at $\xS_i$}:
Each node $\xS_i$ has $n$ $\xE$ node children $\xE_{i1},\dots, \xE_{in}$. Each $\xE_{ij}$ is associated with type $t_j$, and at $\xE_{ij}$ a type-$t_i$ agent is incentivized to report $t_j$. For each $k= 1,\dots, n$, reporting type $t_k$ at $\xE_{ij}$ leads to an $\xS$ node child $\xS_{ijk}$.

\item[iii.] {\bf DIC at $\xS_{ijk}$}:
The signaling strategy used at each node $\xS_{ijk}$ directly recommends the agent an action to take, and the action is optimal for type $t_k$.
\end{itemize}
\end{restatable}

\begin{proof}[Proof sketch]
Given any $\Pi$, we construct an $\EnSES$ mechanism $\Pi'$ that ensures the stated properties while preserving the principal's payoff by following the steps below.

\begin{itemize}
\item 
\emph{Step 1. Adding a new root.}
We make $n$ copies $\Pi^1,\dots,\Pi^n$ of $\Pi$ and attach each $\Pi^i$ as a subtree to the root of $\Pi'$. Let the root node be an $\xEN$ node.
Since all the copies are the same, each type $t_i$ is incentivized to choose $\Pi^i$ and we get IC at the root.

\item 
\emph{Step 2. Removing $\xEN$ nodes.}
Given IC at the root, in each branch $\Pi^i$, we only need to consider the response of $t_i$.
Hence, all the $\xEN$ nodes in $\Pi^i$ become redundant: we remove each $\xEN$ node $e$ and reconnect the child of $e$ selected by a type-$t_i$ agent to the parent of $e$.

\item 
\emph{Step 3. Merging $\xS$ nodes.}
Next, we merge consecutive $\xS$ nodes on every path as if signals generated in these consecutive signaling processes are generated at once: i.e., we treat the combination of the signals as a meta-signal sent by the joint process.
\end{itemize}

The above operations yield a tree with five levels containing only $\xEN$, $\xS$, $\xE$, $\xS$, and leaf nodes, respectively. We proceed as follows to achieve DIC at the $\xS$ nodes.

\begin{itemize}
\item 
\emph{Step 4. Ensuring DIC.}
The idea is to merge signals incentivizing the same response of the agent (i.e., reporting the same type or taking the same action) into a single signal. The signal can be viewed as a recommendation to the agent (in terms of what to report and which action to take).
Essentially, these merging operations will make many nodes in $\Pi$ indistinguishable to the agent, but if the agent follows the recommendations in the new mechanism, they would still reach each leaf with the same probability as before, as if they respond optimally in $\Pi$.
On the other hand, since these merging operations reduce information, the agent cannot devise any better response to improve their payoff. Hence, the agent is incentivized to follow the recommendations and this ensures DIC. 
\qedhere
\end{itemize}
\end{proof}

\subsection{An LP Formulation}

The characterization in \Cref{lmm:EnSEA-optimal} fixes the structure of mechanism. The only parameters to be optimized are the strategies at the $\xS$ nodes, including $\pi_i: \Theta \to \Delta(T)$ at each $\xS_i$, and $\pi_{ijk}:\Theta \to \Delta(A)$ at each $\xS_{ijk}$, where the signal spaces of the strategies are determined by the DIC properties. 
At a high level, this results in an {\em extensive-form game}. We formulate the problem of optimizing the principal's commitment in this game as an LP. The variables of the LP are as follows; each of them corresponds to the marginal probabilities of a path on the game tree, and the feasible space of the variables can be characterized by a set of flow constraints. 
\begin{itemize}
\item 
$\pi_i(t_j \given \theta)$ for $i,j \in \{ 1, \dots, n\}$, which represents the strategy $\pi_i$.

\item 
$\phi_{ijk}(a \given \theta)$ for $i,j \in \{1, \dots, n\}$ and $a \in A$, which represents the product $\pi_i(g_j \given \theta) \cdot \pi_{ijk}(g_{a} \given \theta)$ and indirectly represents $\pi_{ijk}$. We use these variables instead of $\pi_{ijk}(g_{a} \given \theta)$ to avoid quadratic terms in our formulation. 
\end{itemize}

The objective of the LP is to maximize the expected payoff of the principal yielded by the mechanism. Under the IC and DIC conditions, this can be expressed as follows (where we highlight the variables in blue):
\begin{align*}
\max \quad 
\sum_{\theta \in \Theta}\ 
\sum_{i=1}^n \sum_{j=1}^n\ 
\sum_{a \in A}
\rho(t_i) \cdot \mu(\theta) \cdot \varcolor{\phi_{ijj}(a \given \theta)} \cdot v(a, \theta),
\end{align*}
Namely, the expression assume that each type $t_i$ is incentivized to select option $i$ at $\xEN$, report $t_j$ at each $\xE_{ij}$, and play the action recommended by each $\pi_{ijj}$. These conditions are enforced next through the constraints of the LP:
\begin{itemize}
\item 
The flow constraint ensures the product encoded in $\phi_{ijk}$ correspond to a feasible $\pi_{ijk}$:
\begin{align*}
& \sum_{a \in A} \varcolor{\phi_{ijk}(a \given \theta)} = \varcolor{\pi_i(t_j \given \theta)} 
& \text{ for all } 
i,j,k \in \{1,\dots, n\} \text{ and } \theta \in \Theta \\
& \varcolor{\phi_{ijk}(a \given \theta)} \ge 0
& \text{ for all } 
i,j,k \in \{1,\dots, n\},\ a \in A, \text{ and } \theta \in \Theta
\end{align*}
We also have that $\pi_i(\cdot \given \theta) \in \Delta(T)$ for all $i$ and $\theta$, which ensure that $\pi_i$ is a valid strategy.

\item 
The following constraint ensures DIC at each $\xS_{ijk}$, i.e., it is optimal for the {\em imitated type} $t_k$ to play the recommended action $a$, rather than any other action $b$:
\begin{align*}
\sum_{\theta} \mu(\theta) \cdot \varcolor{\phi_{ijk}(a \given \theta)} \cdot u_{t_k}(a, \theta) 
&\ge 
\sum_{\theta} \mu(\theta) \cdot \varcolor{\phi_{ijk}(a \given \theta)} \cdot u_{t_k}(b, \theta)\\
&\qquad\qquad\qquad\qquad
\text{for all } a,b \in A \text{ and } i,j,k \in \{1,\dots, n\}.
\end{align*}
Namely, the left side is the marginal payoff of a type-$t_k$ agent for playing action $a$ when $a$ is recommended, and the right side is that for playing any other action $b$.

\item
Given the above constraint, we introduce the following additional variable $u_{t_i}(\xS_{\ell j k})$ to capture a type-$t_i$ agent's maximum attainable payoff at each $\xS_{\ell j k}$:
\begin{align*}
\varcolor{u_{t_i}(\xS_{\ell j k})} = \sum_{a \in A} \sum_{\theta} \mu(\theta) \cdot \varcolor{\phi_{\ell jk}(a \given \theta)} \cdot u_{t_i}(a, \theta).
\end{align*}
The following constraint then ensures DIC at each $\xE_{ij}$, i.e., it is optimal for an agent whose {\em actual type} is $t_i$ to report the recommended type $t_j$, rather than any other type $t_k$:
\begin{align*}
& \varcolor{u_{t_i}(\xS_{i j j})} \ge \varcolor{u_{t_i}(\xS_{i j k})}
& \text{for all } i,j,k \in \{1,\dots, n\}.
\end{align*}

\item Finally, the following constraint ensures IC at $\xEN$, so that each type $t_i$ is incentivized to select option $i$, rather than any other option $\ell$:
\begin{align}
\label{lp:opt-IC-EN}
& \sum_{j=1}^n \varcolor{u_{t_i}(\xS_{i j j})} \ge \sum_{j=1}^n\ \max_{k=1,\dots, n}\ \varcolor{u_{t_i}(\xS_{\ell j k})}
&\text{for all } i,\ell \in \{1,\dots, n\}.
\end{align}
On the right side, the payoff is yield when the agent selects option $\ell$ and subsequently reports an optimal type $t_k$ at each node $\xE_{\ell j}$; hence, there is a maximization operator over $k$.
As in Section \ref{sec:non-cred}, the constraint can be linearized by further introducing an auxiliary variable $\overline{u}_{t_i}(\xE_{\ell j})$, along with constraints $\overline{u}_{t_i}(\xE_{\ell j}) \ge u_{t_i}(\xS_{\ell j k})$ for all $k$, to capture an upper bound on $\max_{k=1,\dots, n}\ u_{t_i}(\xS_{\ell j k})$.
With this the right side of \Cref{lp:opt-IC-EN} can be replaced with $\overline{u}_{t_i}(\xS_{\ell j k})$.
\end{itemize}

Given the LP, the tractability of optimal indefinite-stage mechanisms then follows readily.

\begin{restatable}{theorem}{thmEnSESLP}
\label{thm:EnSES-LP}
An optimal indefinite-stage mechanism can be computed in polynomial time.
\end{restatable}

\section{Further Improvement: Partial Information Elicitation}
\label{sec:partial-info-elicit}

In the mechanisms we have looked at so far, a binding elicitation directly queries the agent's type. In some scenarios, elicitation might be implemented in a less direct format.
For example, the principal may ask the agent ``are you type $t_1$ or not?'', or present several subsets of types and ask the agent which subset contains the agent's type. After a subset is reported, the principal may follow up with more questions to identify smaller subsets containing the agent's type (that they want to imitate) until eventually pinning down the type in a singleton. Essentially, this splits one canonical elicitation step into multiple steps interleaved with singling steps. We refer to such elicitation as {\em partial information elicitation} (PIE) and mechanisms using partial elicitation as {\em PIE mechanisms}. 

PIE is still {\em binding}, so that once an agent identifies a subset containing his true type, he needs to subsequently report and behave in consistency with this subset.
This is crucial; indeed, such partial elicitation does not provide any added value when the agent is non-credible, as we demonstrated previously. 
In the next example, we demonstrate that PIE can strictly improve the principal's payoff (when the agent is credible). 

\subsection{Strict Payoff Improvement by Partial Elicitation}

\begin{example}
\label{exp:partial-gt-EnSESA}
Let $\Theta = \{\alpha, \beta\}$, $T = \{t_0, t_1, t_2\}$, and $A = \{a,a',b,b'\}$.
The principal gets payoff $1$ for the state-action pairs $(\alpha, a')$ and $(\beta, b')$, and $0$ for all other pairs.
The agent's payoffs are given in the table below, where all blank entries are $-1$.
The state is distributed uniformly: $\mu(\alpha) = \mu(\beta) = 1/2$. The type distribution is: $\rho(t_0) = 1$ and $\rho(t_1) = \rho(t_2) = 0$ (i.e., only $t_0$ actually appears).\footnote{For simplicity, we use zero-probability types in this example. One can also construct a similar but more sophisticated example where all types have non-zero probabilities.}
\begin{center}
\renewcommand{\arraystretch}{1.3}
\small
\begin{tabular}{ L | R | R | R | R |}
         \myheader{} & \myheader{a} & \myheader{a'} & \myheader{b} & \myheader{b'} \\
         \cline{2-5}
         \alpha & 1 & 0 & -2 & -3 \\
         \cline{2-5}
         \beta & -2 & -3 & 1 & 0  \\
         \cline{2-5}
         \mytablecaption{} & \multicolumn{4}{c}{Type $t_0$} 
    \end{tabular}
    \qquad
    \begin{tabular}{ L | R | R | R | R |}
         \myheader{} & \myheader{a} & \myheader{a'} & \myheader{b} & \myheader{b'} \\
         \cline{2-5}
         \alpha &  & 1 &  &  \\
         \cline{2-5}
         \beta &  & 1 & &  \\
         \cline{2-5}
         \mytablecaption{} & \multicolumn{4}{c}{Type $t_1$} 
    \end{tabular}
    \qquad
    \begin{tabular}{ L | R | R | R | R |}
         \myheader{} & \myheader{a} & \myheader{a'} & \myheader{b} & \myheader{b'} \\
         \cline{2-5}
         \alpha &  & &  & 1 \\
         \cline{2-5}
         \beta &  &  &  & 1  \\
         \cline{2-5}
         \mytablecaption{} & \multicolumn{4}{c}{Type $t_2$} 
    \end{tabular}
\end{center}
\end{example}

\subsubsection{With PIE Mechanism}

Consider the following PIE mechanism. The principal first asks the agent whether their type is $t_0$. 
\begin{itemize}
\item 
If the answer is ``yes'', then no information will be sent subsequently, and the agent decides an action to take. (By credibility, this action needs to be optimal for $t_0$.)

\item 
If the answer is ``no'', then the principal will send a signal that reveals the true state, and then ask the agent to identify their type from the set $T \setminus \{t_0\}$ (given that the agent has denied that their type is $t_0$). Finally, the agent performs an optimal action of the type they selected.
\end{itemize}

It can be verified that the above mechanism yields payoff $1$ for the principal, which is the maximum possible payoff of the principal in this example.
Specifically (note that we only need to consider a type-$t_0$ agent in this case because the other types have probability $0$):
\begin{itemize}
\item 
A type-$t_0$ agent will deny that their type is $t_0$ in response to the principal's first question: If they reported $t_0$, no information would be provided, in which case states $\alpha$ and $\beta$ would still be equally likely according to their belief. Consequently, the agent's expected payoff would be negative for all possible actions.

\item 
Subsequently, when the principal reveals $\alpha$, the agent will be incentivized to report and act as $t_1$, responding with $t_1$'s optimal action $a'$. 
When the principal reveals $\beta$, the agent will act as $t_2$ and respond with $t_2$'s optimal action $b'$. 
In both states, the agent gets payoff $0$ (strictly higher than the above case), and the principal gets payoff $1$.
\end{itemize}

\subsubsection{With Non-PIE Mechanism}
Now consider an arbitrary non-PIE (indefinite-stage) mechanism $\Pi$.
We show that $\Pi$ yields payoff lower than $1$ for the principal.

\begin{lemma}
Any non-PIE mechanism $\Pi$ yields payoff lower than $1$ for the principal in \Cref{exp:partial-gt-EnSESA}.
\end{lemma}

\begin{proof}
Suppose for the sake of contradiction that $\Pi$ yields payoff $1$ for the principal. 
W.l.o.g. $\Pi$ is an $\EnSES$ mechanism as we argued in \Cref{sec:compute}.

Since in the example the probability $t_1$ and $t_2$ actually appear is zero, only a type-$t_0$ agent would contribute to the principal's payoff.
Consider the distribution induced by $\Pi$ and a type-$t_0$ agent's optimal response over tuples $(\theta, t) \in \Theta \times T$, where $t$ is the type the agent reports.
Let $(\theta, t)$ be an arbitrary tuple with a positive probability in this distribution.
It must be that $t \neq t_0$. Indeed, if the agent acts according to $t_0$, they will play $a$ or $b$ as the other actions are strictly dominated. Since the principal's payoff is $0$ for both $a$ and $b$, the principal will get $0$ with a positive probability, contradicting the assumption that $\Pi$ yields payoff $1$ for the principal.

As a result, we have $t \in \{t_1, t_2\}$.
If $t = t_1$, it must be $\theta = \alpha$. Indeed, a type-$t_1$ agent can only be incentivized to play $a'$ according to the payoff definition, whereas the principal's payoff for $(\beta, a')$ is $0$.
Consider further the binding elicitation step (which is not partial elicitation as $\Pi$ is non-PIE) and the agent's posterior belief $p \in \Delta(\Theta)$ at this step. There are two possible cases:
\begin{itemize}
\item $p(\alpha) = 1$. In this case, a type-$t_0$ agent  would in fact be better off reporting $t_0$ and playing $a$, instead of reporting $t_1$ (and playing $a'$); the payoff is $1$ for the former and $0$ for the latter.
Hence, $p(\alpha) = 1$ is not possible.

\item $p(\alpha) < 1$ (and hence $p(\beta) > 0$).
In this case, since $p(\beta) > 0$, with a positive probability the state-action pair $(\beta, a')$ will be realized (given the assumption that the agent will act as $t = t_1$ subsequently). This means that the principal will receive payoff $0$ with a positive probability, which contradicts the assumption that $\Pi$ yields payoff $1$ for the principal, too.
\end{itemize}

By symmetry, the case where $t = t_2$ also leads to contradictions. The assumption that $\Pi$ yields payoff $1$ for the principal does not hold.
\end{proof}

The above result and the PIE mechanism we presented above establish the following result.

\begin{theorem}
There exists an instance where an optimal PIE mechanism yields a strictly higher payoff for the principal than any non-PIE mechanisms do.
\end{theorem}

\subsection{Suboptimality of Constant-Step PIE Mechanisms}

It might be tempting to think that, similarly to the case of non-PIE mechanisms, it is without loss of optimailty to consider PIE mechanisms with a constant depth. Surprisingly, as we demonstrate using the following example, this is not the case: in this example, the depths of {\em all} the optimal PIE mechanisms grow with the size of the instance. 

\begin{example}
\label{exp:PIE-n-steps}
Let $\Theta = \{\theta_1, \dots, \theta_n \}$, $T = \{t_0, t_1, \dots, t_n\}$, and $A = \{a_1,\dots, a_n\} \cup \{a'_1, \dots, a'_n\}$.
The principal gets payoff $1$ for the state-action pairs $(\theta_i, a'_i)$ for all $i = 1,\dots, n$, and gets $0$ for all other pairs.
The agent's payoffs are given in \Cref{tab:pie_beyond_const}, where we let $M = 2n^{3^n}$. The parameters are non-negative only on the diagonal from $(\theta_1, a_1)$ to $(\theta_n, a_n)$ and in the entry $(\theta_n, a'_n)$.
The state is distributed uniformly: $\mu(\theta_1) =  \dots = \mu(\theta_n) = 1/n$. The type distribution is: $\rho(t_0) = 1$ and $\rho(t_1) = \dots = \rho(t_n) = 0$ (i.e., only $t_0$ actually appears).
\begin{table}[t]
\renewcommand{\arraystretch}{1.3}
    \centering
    \small
    \begin{tabular}{LCCCCCCCCCC}
        \myheader{} & \myheader{a_1} & \myheader{a_2} & \myheader{\dots} & \myheader{a_{n-1}} & \myheader{a_n} & \myheader{a'_1} & \myheader{a'_2} & \myheader{\dots} & \myheader{a'_{n-1}} & \myheader{a'_n} \\
         \cline{2-11}
         \multicolumn{1}{L|}{\theta_1} & 1 \tikzmark{nw1} &   &  &  & \multicolumn{1}{L|}{}  & 0 \tikzmark{nw2} &   &  &  & \multicolumn{1}{L|}{}  \\
\multicolumn{1}{L|}{\theta_2} &   & 1 &  & \multicolumn{1}{r}{\Large \hspace{-5mm}$0$\hspace{-5mm}} & \multicolumn{1}{L|}{}  &   & 0 &  & \multicolumn{1}{r}{\large \hspace{-5mm}$-M$\hspace{-5mm}} & \multicolumn{1}{L|}{}  \\
\multicolumn{1}{L|}{\vdots}       &  &  & \ddots &        & \multicolumn{1}{L|}{} &  &  & \ddots &   & \multicolumn{1}{L|}{} \\
\multicolumn{1}{L|}{\theta_{n-1}} &  & \multicolumn{1}{c}{\Large \hspace{-5mm}$0$\hspace{-5mm}} &        & 1 & \multicolumn{1}{L|}{} &  & \multicolumn{1}{c}{\large \hspace{-5mm}$-M$\hspace{-5mm}} &        & 0 & \multicolumn{1}{L|}{} \\
\multicolumn{1}{L|}{\theta_n} &   &   &  &  & \multicolumn{1}{L|}{1 \tikzmark{se1}} &   &   &  &  & \multicolumn{1}{L|}{1 \tikzmark{se2}} \\
         \cline{2-11}
         \mytablecaption{} & \multicolumn{10}{c}{\em Type $t_0$}
    \end{tabular}
    \quad
    \begin{tabular}{ L | C | C | C | C | C | C |}
         \myheader{} & \myheader{a_1} & \myheader{\dots} & \myheader{a_i} & \myheader{\dots}  & \myheader{a'_i} & \myheader{\dots} \\
         \cline{2-7}
         \theta_1& \cellcolor{lightgray}0 & &  & & &  \\
         \cline{2-7}
         \vdots &  & \cellcolor{lightgray}\ddots & & & &  \\
         \cline{2-7}
         \theta_i &  &  & \cellcolor{lightgray}0 &  & \multicolumn{1}{c|}{\cellcolor{lightgray} \footnotesize \!\!\!\!$1/M$\!\!\!\!} &  \\
         \cline{2-7}
         \vdots & & & & & & \\
         \cline{2-7}
         \mytablecaption{} & \multicolumn{6}{c}{\em Type $t_i$}
    \end{tabular}

    \begin{tikzpicture}[overlay,remember picture, shorten >=-3pt]
    \draw[-,dotted] ([xshift=-4.8mm,yshift=-1.5mm]pic cs:nw1) -- ([xshift=-7mm,yshift=-1.1mm]pic cs:se1);
    \draw[-,dotted] ([xshift=2.5mm,yshift=3.8mm]pic cs:nw1) -- ([xshift=3.1mm,yshift=2.4mm]pic cs:se1);
    \draw[-,dotted] ([xshift=-4.8mm,yshift=-1.5mm]pic cs:nw2) -- ([xshift=-7mm,yshift=-1.1mm]pic cs:se2);
    \draw[-,dotted] ([xshift=2.5mm,yshift=3.8mm]pic cs:nw2) -- ([xshift=3.1mm,yshift=2.4mm]pic cs:se2);
    \end{tikzpicture}

    \caption{The agent's payoffs in \Cref{exp:PIE-n-steps}. The payoffs of an agent of type $t_i$ (for each $i = 1,\dots,n$) is the same as those of type $t_0$, with the exception of the $i+1$ entries highlighted in gray (payoffs in the blank entries are all the same as type $t_0$).}
    \label{tab:pie_beyond_const}
\end{table}

\end{example}

\subsubsection{An Optimal PIE Mechanism}
\label{sec:opt-pie}

An $n$-step PIE mechanism is illustrated in \Cref{fig:n-step-PIE}.
The mechanism begins with an $\xE$ node, $\xE_0$, and asks the agent whether they are of type $t_i$ or not at each subsequent node $\xE_i$: answering ``yes'' leads to termination of the mechanism, while answering ``no'' leads to continuation of the mechanism to node $\xS_{i+1}$.
At each $\xS_{i}$, the signaling strategy signals $g_i$ deterministically if the state is $\theta_i$, and signals $g_{>i}$, otherwise.

Recall that $t_0$ is the only type that appears with a positive probability. It can be verified that a type-$t_0$ agent is (weakly)\footnote{One can change the $0$'s on the diagonal from $(\theta_1, a'_1)$ to $(\theta_{n-1}, a_{n-1})$ to a sufficiently small number to strongly incentivize the agent and, with extra bookkeeping, prove that the example still works.} incentivized to respond as follows:
\begin{itemize}
\item Act as $t_i$ and play action $a'_i$ upon receiving $g_i$, whereby he receives payoff $0$. Indeed, according to the construction, we have $\Pro(\theta_i \given g_i) = 1$, so the type-$t_0$ agent could only be better off playing $a_i$. However, playing $a_i$ is blocked by the credibility constraint because upon reaching $\xS_i$ it must be that the agent has denied that they are of types $t_0, \dots, t_{i-1}$; playing $a_i$ cannot be justified as an optimal action of any of the remaining types in this case.

\item Report $\neg t_i$ at each $\xE_i$.
This can be verified by noting that the agent's posterior belief at $\xE_i$ is a uniform distribution over $\theta_{i+1}, \dots, \theta_n$. Any of $a_{i+1}, \dots, a_n$ is optimal for the type-$t_0$ agent and gives payoff $1/(n-i)$. In comparison, if the agent keeps reporting $\neg t_j$ at the each subsequent $\xE_j$ and acts as type $t_j$ upon receiving $g_j$ as described above, he will be persuaded to play precisely $a'_j$ at each $\theta_j$, whereby he gets the same overall payoff $1/(n-i)$ (thanks in particular to the payoff $1$ for $(\theta_n, a'_n)$). Hence, it is optimal to report $\neg t_i$ at each $\xE_i$. 
\end{itemize}
Since action $a'_i$ is played at every $\theta_i$, the principal always gets payoff $1$, which is also the maximum possible payoff in this example. Hence, the PIE mechanism is optimal.

Intuitively, the PIE mechanism works by revealing a state $\theta_i$ only after the agent denies his type being any of $t_1, \dots, t_{i-1}$. This is crucial as otherwise the agent would have been better off acting as any arbitrary type in this range, whose optimal action is $a_i$. Moreover, eliciting the exact type immediately at any $\xE_i$ would not work as that would only encourage the agent to report $t_i$ and play one of the actions $a_{i+1},\dots, a_n$, leading to payoff $0$ for the principal. We next formally analyze the lower bound on the depths of optimal PIE mechanisms.

\begin{figure*}
\centering
\begin{tikzpicture}
    \node[E node] (1) {$\xE_i$};
    \node (2) [below left = 1cm of 1] {};
    \node[S node, text width=6mm] (3) [below right = 1cm of 1] {$\xS_{i+1}$};

    \draw[arrleaf] (1) -- node [black, midway, above left] {$t_i$} (2) [anchor=west];
    \draw (1) -- node [black, midway] {$\lnot t_i$} (3) [anchor=west];
\end{tikzpicture}
\qquad\qquad\qquad
\begin{tikzpicture}
    \node[S node] (1) {$\xS_{i}$};
    \node (2) [below left = 1cm of 1] {};
    \node[E node] (3) [below right = 1cm of 1] {$\xE_{i}$};

    \draw[arrleaf] (1) -- node [black, midway, above left] {$g_i$} (2) [anchor=west];
    \draw (1) -- node [black, midway] {$g_{>i}$} (3) [anchor=west];
\end{tikzpicture}
\caption{A optimal mechanism for \Cref{exp:PIE-n-steps}, $i \in \{0, 1, \dots, n-1\}$. $\xE_n$ is a terminal node.}
\label{fig:n-step-PIE}
\end{figure*}
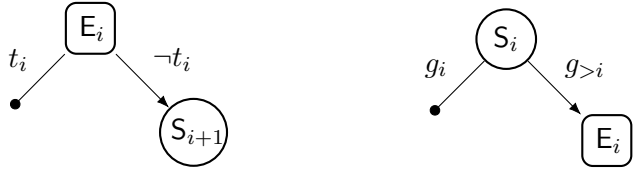

\subsubsection{Depth of Optimal PIE Mechanisms}

\newcommand{\subsetsim}{\mathrel{\ooalign{\raise.3ex\hbox{$\subset$}\cr$\hspace{0.2ex}\raise-.7ex\hbox{$\sim$}$}}}
\newcommand{\children}{\mathrm{children}}
\newcommand{\IT}{\widehat{T}}
\newcommand{\ITheta}{\widehat{\Theta}}
\newcommand{\calP}{\mathcal{P}}

We demonstrate an $\Omega(n)$ lower bound on the depths of the optimal PIE mechanisms, for a class of instances with $n$ possible states.
Consider an arbitrary optimal PIE mechanism $\Pi$. Given the PIE mechanism we presented above, for $\Pi$ to be optimal, it must yield payoff $1$ for the principal. We will show that to achieve this payoff, $\Pi$ must have $\Omega(n)$ levels.

Since type $t_0$ is the only type that actually appears, an argument similar to the proof of \Cref{lmm:EnSEA-optimal} implies that there is no need for $\Pi$ to use any $\xEN$ node.
We can then assume w.l.o.g. that $\Pi$ starts with an $\xE$ node, and subsequently the levels alternate between $\xS$ and $\xE$ nodes (with the exception of leaf nodes, which can appear at any level). 
For each node $x$ in $\Pi$, we define the following sets:
\begin{itemize}
\item $T_x \subseteq T$, which contains the types that are {\em not} excluded by the agent's answers at the $\xE$ nodes on the path from the root to $x$.

\item $\Theta_x \subseteq \Theta$, which is the support of the agent's posterior at $x$, i.e., $\theta \in \Theta$ if and only if $\Pro(\theta \given x) > 0$, where $\Pro$ is the probability measure induced by $\Pi$. 
\end{itemize}
To simplify the notation, we write $p_x$ as the distribution such that $p_x(\theta) := \Pro(\theta \given x)$ for each $\theta$.

The next two lemmas follow by noting that to yield payoff $1$ for the principal requires each $a'_i$ to be played exactly at state $\theta_i$, and this behavior can be induced only when the agent acts as type $t_i$. 

\begin{restatable}{lemma}{lmmleafTx}
\label{lmm:leaf-T-x}
$|\Theta_x| = 1$ for every leaf node $x$ in $\Pi$.
\end{restatable}

\begin{restatable}{lemma}{lmmThetaxTx}
\label{lmm:Theta-x-T-x}
For every node $x$ in $\Pi$, if $\theta_i \in \Theta_x$ then $t_i \in T_x$.
\end{restatable}

Moreover, we can prove the following key result to establish a connection between posteriors at each node and their child nodes.

\begin{restatable}{lemma}{lmmchildinPi}
\label{lmm:child-in-Pi}
For every $i=1,\dots, n-1$ and $\epsilon \in \mathbb{R}$, let $\calP_i^\epsilon \subseteq \Delta(\Theta)$ consist of points $p$ such that: 
\begin{itemize}
\item[i.] 
$\sum_{j = i}^n p(\theta_j) =1$;

\item[ii.] 
$\sum_{j = i+1}^n p(\theta_j) \ge 1 - \frac{1}{n-i+1} - \epsilon$; and

\item[iii.] $\left| p(\theta_j) - p(\theta_n) \right| \le \epsilon/2$ for all $j = i+1, \dots, n$.
\end{itemize}
Let $\epsilon_k = n^{-3^{n-k}}$ for $k \in \mathbb{Z}$.
For every non-leaf node $x$ in $\Pi$, if $t_i \in T_x$ and $p_x \in \calP_i^{\epsilon_k}$, then there exists $c \in \children(x)$ such that $p_c \in \calP_i^{\epsilon_{k+1}}$.
\end{restatable}

Intuitively, when $\epsilon$ is small, $\calP_i^\epsilon$ is roughly a hyperplane in $\Delta(\Theta)$, where $p(\theta_{i+1}) = \dots = p(\theta_n)$. 
The above lemma states that if type $t_i$ is not yet excluded at an $\xS$ node, then at the next level of the tree there must be a node $c$ at which the posterior stays in $\calP_i^{\epsilon_{k+1}}$, which is roughly the same as $\calP_i^{\epsilon_k}$.
When $\epsilon_{k+1}$ is sufficiently small, this further means that the posterior $p_c$ remains supported on $\{\theta_{i+1}, \dots, \theta_n\}$ because the probabilities of these states cannot be too close to $0$ due to the first condition in the lemma.
By~\Cref{lmm:leaf-T-x}, the mechanism cannot terminate at $c$ in this case, so $\Pi$ needs to include one more level.
Applying this argument repeatedly we find that $\Pi$ must have $\Omega(n)$ levels.

\begin{theorem}
\label{prp:depth-n}
There exists a class of instances where every optimal PIE mechanism has depth $\Omega(n)$.
\end{theorem}

\begin{proof}
Consider the instance given by \Cref{exp:PIE-n-steps} and the following induction process.
\begin{itemize}
\item 
Since $\mu$ is uniform in the example, at the root node $r$ we have $p_r = \mu \in \calP_1^{\epsilon_1}$ and $t_1 \in T_r = T$.

\item 
Suppose that at some non-leaf node $x$, it holds that $p_x \in \calP_i^{\epsilon_k}$ and $t_i \in T_x$ for $i,k < n - 1$.
By \Cref{lmm:child-in-Pi}, $p_c \in \calP_i^{\epsilon_{k+1}}$ for some $c \in \children(x)$.
Consider the following two cases.
\begin{itemize}
\item 
If $p_c(\theta_i) > 0$, which means $\theta_i \in \Theta_c$, then by \Cref{lmm:Theta-x-T-x} it must be $t_i \in T_c$. Hence, there exists a child $c$ such that $p_c \in \calP_i^{\epsilon_{k+1}}$ and $t_i \in T_c$.

\item 
If $p_c(\theta_i) = 0$, then since $p_c \in \calP_i^{\epsilon_{k+1}}$, by definition we have $\sum_{j=i+1}^n p_c(\theta_j) = 1$.

In this case, it must be that 
\begin{align}
\label{eq:depth-n-pc}
0 < p_c(\theta_{i+1}) \le \frac{1}{n-i} + \frac{\epsilon_{k+1}}{2}.
\end{align}
Indeed, if $p_c(\theta_{i+1})$ was larger than the upper bound, we would have 
$p_c(\theta_j) > 1/(n-i)$
for all $j = i+1, \dots, n$ due to condition (iii) in the definition of $\calP_i^{\epsilon_{k+1}}$; this leads to the contradiction $\sum_{j=i+1}^n p_c(\theta_j) > 1$.
Similarly, $p_c(\theta_{i+1})=0$ would imply $p_c(\theta_j) < \epsilon_{k+1}$ by condition (iii), and in turn the contradiction $\sum_{j=i+1}^n p_c(\theta_j) < n \cdot \epsilon_{k+1} < 1$.
 
The upper bound implies 
\begin{align*}
\sum_{j=i+2}^n p_c(\theta_j) 
= 1 - p_c(\theta_{i+1}) 
\ge 1 - \frac{1}{n-i} - \epsilon_{k+1}.
\end{align*}
Since $p_c \in \calP_i^{\epsilon_{k+1}}$ already implies that $|p_c(\theta_j) - p_c(\theta_n)| \le \epsilon_{k+1}/2$ for all $j = i+2, \dots, n$.
This means 
$p_c \in \calP_{i+1}^{\epsilon_{k+1}}$.

Additionally, the lower bound in \Cref{eq:depth-n-pc} implies $\theta_{i+1} \in \Theta_c$, and in turn $t_{i+1} \in T_c$ by \Cref{lmm:Theta-x-T-x}.
\end{itemize}

Therefore, in both cases, we find a child $c$ such that $t_\ell \in T_c$ and $p_c \in \calP_{\ell}^{\epsilon_{k+1}}$ for some $\ell \le i + 1$.
Since $i < n-1$, we also have $|\Theta_c| > 1$ by noting that $p_c(\theta_j) > 0$ for all $j=i+2,\dots,n$ for the same reason that $p_c(\theta_{i+1}) > 0$ as we argued in \Cref{eq:depth-n-pc}. Hence, $c$ is not a leaf node by \Cref{lmm:leaf-T-x}.
\end{itemize}

By induction, there exists a non-leaf node at the $(n-2)$-th level. Hence, the depth of $\Pi$ is $\Omega(n)$.
\end{proof}

The lower bound is tight by a similar argument to \Cref{lmm:EnSEA-optimal} and noting that it is w.l.o.g. to design partial elicitation in a way such that $T_c \subsetneq T_x$ for any $\xE$ node $x$ and $c \in \children(x)$. We state the upper bound below and omit the proof. 

\begin{theorem}
\label{thm:PIE-depth-ub}
For any PIE mechanism $\Pi$, there exists a PIE mechanism $\Pi'$ with depth $O(n)$ that yields as much payoff for the principal as $\Pi$ does.
\end{theorem}

We remark that the size of the tree of an optimal mechanism may still grow exponentially in its depth. Because of this, we cannot use LP formulations similar to the one in \Cref{sec:compute} to obtain efficient algorithms for computing optimal PIE mechanisms.
We leave open this problem.
The problem is intriguing in that even when a single partial elicitation node $x$ is considered, there are exponentially many possible ways to set up the question for the agent as every subset of $T_x$ may be involved.

\section{Conclusion}
We studied Bayesian persuasion with credible agents.
For two-stage---elicitation then signaling---mechanisms, we showed that IC policies are easy to compute but may be suboptimal, while the optimal policy may be NP-hard to compute.
We demonstrate that two additional instruments---pre-signaling and non-binding
elicitation---lead to mechanisms involving additional stages, can strictly improve the principal’s utility, are necessary for achieving optimality, and can make the problem computationally tractable.
Finally, we also show that a specific form of elicitation that we call partial information elicitation can further improve the principal’s utility, but an unbounded number of rounds of information exchange between the principal and the agent may be necessary to achieve optimality.
Future directions include analyzing the computational complexity of the optimal PIE mechanisms, which we leave open in the current paper.
It would also be interesting to consider credible agents in similar problems beyond persuasion.


\clearpage


\appendix

\section{Omitted Proofs}

\subsection{Proof of \Cref{thm:nph-no-IC}}

\thmnphnoIC*

\begin{proof}
We give a tight reduction from {\sc Maximum Independent Set}.
Let $G = (V, E)$ be an undirected graph, an arbitrary instance of {\sc Maximum Independent Set}. Let $n = |V|$ denote the number of vertices. Let $N(v) = \{ v' \given (v, v') \in E \}$ denote the set of neighbors of $v$. It is known that approximating the {\sc Maximum Independent Set} problem better than a factor of $1/n^{1-\epsilon}$ for any $\epsilon > 0$ is NP-hard~\cite{zuckerman2006linear}, i.e., 
there is no polynomial-time algorithm, unless P=NP, which, given an arbitrary undirected graph $G$ with a maximum independent set of size of $k$, can always tell whether $G$ has a maximum independent set of size larger or smaller than $k/n^{1-\epsilon}$. 

The reduction goes as follows:
Let $T = \{t_v : v \in V\} \cup \{t_*\}$, i.e., there is a type $t_v$ corresponding to each vertex $v \in V$ and a special type $t_*$. Note that $|T| = n+1$.
In the distribution over types $\rho$, each $t_v$ has an equal probability of occurring, $\rho(t_v) = 1/n$, but the type $t_*$ has a zero probability, $\rho(t_*) = 0$.\footnote{If this is unsatisfactory, it is easy to check that the proof works for $\rho(t_*) = \delta > 0$ as well, for small $\delta$, with more bookkeeping.}
Let there be two states of nature, $\Theta = \{\theta_1, \theta_2\}$, which can occur with equal probability, $\mu(\theta_1) = \mu(\theta_2) = 1/2$.
The action set of the agent is $A = \{a_v : v \in V\} \cup \{b_v : v \in V\} \cup \{a_0, b_0\}$.
The utility function of type $t_v$ is given in the table below, where $v' \in N(v)$ is any arbitrary neighbour of $v$, and $v'' \notin N(v)$ is any arbitrary vertex that is not adjacent to $v$.
\begin{center}
\renewcommand{\arraystretch}{1.3}
\small
\begin{tabular}{L | R | R | R | R | R | R |} 
 \myheader{}
  & \myheader{a_0} & \myheader{b_0} & \myheader{a_v} & \myheader{b_v} & \myheader{a_{v'}, b_{v'}} & \myheader{a_{v''}, b_{v''}} \\ 
 \cline{2-7}
 \theta_1 & 1+1/n & -1+1/n & 1 & -1 & -n & 0 \\ 
 \cline{2-7}
 \theta_2 & -1+1/n & 1+1/n & -1 & 1 & -n & 0 \\ 
 \cline{2-7}
\end{tabular}
\vspace{2mm}
\end{center}
The utility of type $t_*$ is $-1$ for actions $a_0$ and $b_0$, irrespective of the state, and $1$ for all other actions and states.
Irrespective of the state, the payoff of the principal is $0$ for $a_0$ and $b_0$, and $1$ for all other actions.

For each type $t_v$, the better action in $\{a_0, b_0\}$ strictly dominates all other actions. In particular, $a_0$ dominates $a_v$, $a_{v'}$, and $b_{v'}$; $b_0$ dominates $b_v$ (and $a_{v'}$ and $b_{v'}$); for any posterior, say $(p, 1-p)$, of the agent after getting the signal from the principal, the better action in $\{a_0, b_0\}$ gets a utility of $\max(2p - 1 + 1/n, 1 - 2p + 1/n) \ge 1/n > 0$ and dominates $a_{v''}$ and $b_{v''}$. 
So, if the agent reports his type as $t_{v}$, for any $v \in V$, the principal gets a reward of $0$.

On the other hand, type $t_*$ dislikes playing actions $a_0$ and $b_0$, which give him a reward of $-1$, but is indifferent between playing any of the other actions, all of which give him a reward of $1$.
So, the principal gets a reward of $1$ if any agent reports his type as $t_*$.

Essentially, the principal's overall utility depends upon the number of types $t_v$ that can be persuaded to report and behave like $t_*$.
We prove next that the principal's utility is proportional to the size of the maximum independent set in $G$.
Let us first prove that the principal's optimal utility is at least the size of the maximum independent set divided by $n$.
Let $V' \subseteq V$ be any maximum independent set of $G$.
We show that the following signaling policy yields a utility of $|V'|/n$ for the principal.
\begin{itemize}
\item 
The agent reports a type $t_v$ for some $v \in V$. In this case, reveal no information.
\item 
The agent reports $t_*$.
If the state is $\theta_1$,
pick a $v' \in V'$ uniformly at random and recommend action $a_{v'}$, 
i.e., each action $a_{v'}$ for $v' \in V'$ is recommended with probability $1/|V'|$ and other actions are recommended with probability $0$.
On the other hand, if the state is $\theta_2$, recommend each action $b_{v'}$ for $v' \in V'$ with probability $1/|V'|$ and other actions with probability $0$.
\end{itemize}
Given this signaling policy, let us compute the expected reward that an agent of type $t_v$ gets by reporting specific types in $T$:
\begin{itemize}
    \item 
    By reporting $t_{v'}$ for any $v' \in V$, the agent gets no information from the principal. His posterior is the same as the prior $(1/2, 1/2)$, and he gets an expected utility of $1/n$ by playing either $a_0$ or $b_0$.
    \item 
    By reporting $t_*$, the agent gets recommendation $a_{v'}$ or $b_{v'}$ depending upon the state, for some $v' \in V'$. If the vertex $v$ corresponding to the agent's type $t_v$ is in $V'$, then none of the vertices in $V'$ are neighbors of $v$, so the agent gets an expected reward of $1$. On the other hand, if $v \notin V'$, then there must be at least one vertex in $V'$ who is a neighbor of $v$, so the agent has an expected utility of $\le -n/|V'| + (|V'|-1)/|V'| < 0$.
\end{itemize}
The analysis above tells us that a type-$t_v$ agent prefers to report $t_*$ if and only if $v \in V'$. So, the principal gets an expected utility of $|V'|/n$.

Conversely, suppose that the size of the maximum independent sets is $k$. We show that the maximum payoff the principal can achieve is at most $k/n$.
Suppose, for the sake of contradiction, the principal obtains $k'/n$ by using some signaling policy $\pi$, where $k' > k$. Since the principal gets $0$ for each type $t_v$ unless they misreport $t_*$, there must be at least $k+1$ of these types who have been persuaded to misreport $t_*$.
Given that the maximum independent set has size $k$, there must be two adjacent vertices $v_1, v_2 \in V$, whose corresponding types $t_{v_1}$ and $t_{v_2}$ would both misreport $t_*$.
We show that this lead to a contradiction.

Note that by reporting $t_v$, for any $v \in V$, and playing $a_0$ or $b_0$, any type can get a reward of $1/n$. Hence, to incentivize type-$t_{v_1}$ to misreport $t_*$, the state--action pairs $(\theta_1, a_{v_1})$ and $(\theta_2, b_{v_1})$ must be selected with a total probability of at least $1/n$ because these pairs give type-$t_{v_1}$ a utility of $1$ while all other pairs give a utility of $\le 0$. 
Since $v_1$ and $v_2$ are neighbors, type-$t_{v_2}$ gets a negative payoff of $-n$ when these pairs are selected, which results in an overall payoff of at most $-n \cdot \frac{1}{n} + \left(1 + \frac{1}{n} \right) \cdot \left( 1 - \frac{1}{n} \right) < 0 < \frac{1}{n}$ (where $1+\frac{1}{n}$ is the highest attainable payoff for the agent among any action--state pairs). This contradicts the assumption that $t_{v_2}$ is incentivized to misreport $t_*$.

We showed that the optimal expected utility of the principal and the size of the maximum independent set are proportional (equal, except for multiplicative factor of $1/n$). So, if we had a polynomial-time algorithm for approximating principal's optimal expected utility, we could use it to approximate the size of the maximum independent set, which is not possible unless P=NP.
\end{proof}

\subsection{Proof of \Cref{lmm:EnSESA-gt-SESA-option-i}}

\EnSESAoptioni*

\begin{proof}
Consider {\em a type-$t_1$ agent} and compare the agent's payoffs for choosing options~$1$ and $2$. 

If the type-$t_1$ agent selects option~$1$:
\begin{itemize}
\item 
Upon receiving $g_{\alpha,\gamma}$, the agent's posterior belief is $\Pro(\cdot \given g_{\alpha,\gamma}, \xS_1) = \left(\frac{1}{2 - \delta},\, 0,\, \frac{1 - \delta}{2 - \delta} \right)$ (almost evenly distributed over $\alpha$ and $\gamma$). It is then optimal for the agent to take action $a'$, whereby she gets a strictly positive payoff while all other actions gives at most $0$. 

\item 
Upon receiving $g_{\beta,\gamma}$, the agent's posterior belief is $\Pro(\cdot \given g_{\beta,\gamma}, \xS_1) = \left(0,\, \frac{1}{1 + \delta},\, \frac{\delta}{1 + \delta} \right)$. It is optimal for the agent to report (and imitate) $t_2$ in response to the elicitation at $\xE_1$, whereby the agent would play $b'$ upon receiving $g_\beta$ (when the state can only be $\beta$) and $d$ upon receiving $g_\gamma$ (when the state can only be $\gamma$), as if she is of type $t_2$. 
Overall, the type-$t_1$ agent's expected payoff conditioned on $\xE_1$ is $\frac{2}{1 + \delta} > 0$. 
In comparison, if she reports $t_1$ and receives no further information to update her belief, the best action would be $d$, which only gives payoff $0$. 
\end{itemize}
Overall the agent's expected payoff for selecting option~$1$ is {\em at least} 
\[
 \underbrace{\frac{2 - \delta}{3}}_{\Pro(g_{\alpha,\gamma} \given \xS_1)} \times 0 
 \quad + \quad 
 \underbrace{\frac{1 + \delta}{3}}_{\Pro(g_{\beta,\gamma} \given \xS_1)} \times \frac{2}{1 + \delta} 
 \quad = \quad 
 \frac{2+ 2\delta}{3(1 + \delta)}.
\]

If the type-$t_1$ agent selects option~$2$:
\begin{itemize}
\item 
Upon receiving $g_{\alpha,\gamma}$, the agent's posterior belief is $\Pro(g_{\alpha,\gamma} \given \xS_2) = \left(\frac{1}{1 + \delta},\, 0,\, \frac{\delta}{1 + \delta} \right)$.
Then, truthful reporting is obviously optimal at $\xE_2$, in particular because more information will be provided if the agent reports $t_1$. The agent further plays action $a'$ upon receiving $g_\alpha$ (when the state can only be $\alpha$), and $d$ upon receiving $g_\gamma$ (when the state can only be $\gamma$).
The overall payoff conditioned on  $\xE_2$ is $\frac{1}{1 + \delta}$.

\item 
Upon receiving $g_{\beta,\gamma}$, the agent's posterior belief is $\Pro(g_{\beta,\gamma} \given \xS_2) = \left(0,\, \frac{1}{2 - \delta},\, \frac{1 - \delta}{2 - \delta} \right)$. 
No further information will be provided subsequently. The optimal action to play is $d$ and the agent's conditional payoff is $0$.

\end{itemize}
Overall the agent's expected payoff for selecting option~$2$ is {\em at most} 
\[
\underbrace{\frac{2 - \delta}{3}}_{\Pro(g_{\alpha,\gamma} \given \xS_2)} \times \frac{1}{1 + \delta}
\quad + \quad 
\underbrace{\frac{1 + \delta}{3}}_{\Pro(g_{\alpha,\gamma} \given \xS_2)} \times 0
\quad = \quad 
\frac{2-\delta}{3(1 + \delta)},
\]
which is smaller than the payoff obtained by selecting option~$1$.
Hence, a type-$t_1$ agent would strictly prefer option~$1$.

By symmetry, the same analysis can be applied to show that a type-$t_2$ agent would strictly prefer option~$2$. This completes the proof. 
\end{proof}

\subsection{Proof of \Cref{lmm:u-g}}

\lmmug*

\begin{proof}
Suppose for a contradiction that for every $g \in G$, at least one of $\mathbb{E}(v \given g) > 1 - \delta - \epsilon$ and $\Pro(\gamma \given g) \ge 1/4 - \delta/\epsilon$ is violated.
Let $G = G^- \cup G^+$, where 
\[
G^- = \{g \in G : \Pro(\gamma \given g) < 1/4 - \delta/\epsilon\}
\quad\text{and}\quad
G^+ = \{g \in G : \Pro(\gamma \given g) \ge 1/4 - \delta/\epsilon\}.
\]
Note that by definition,
\[
\sum_{g \in G} \Pro(g) \cdot \Pro(\gamma \given g) = \Pro(\gamma) = 1/4.
\]
We have
\begin{align*}
1/4 
&\ =\  
\sum_{g \in G^-} \Pro(g) \cdot \Pro(\gamma \given g) + \sum_{g \in G^+} \Pro(g) \cdot \Pro(\gamma \given g) \\
&\ <\  
(1/4 - \delta/ \epsilon) \cdot \sum_{g \in G^-} \Pro(g) +  \sum_{g \in G^+} \Pro(g) 
\ =\  
(1/4 - \delta/ \epsilon) \cdot ( 1 - \Pro(G^+) ) + \Pro(G^+),
\end{align*}
where we write $\Pro(G^+) = \sum_{g \in G^+} \Pro(g)$.
It follows that
\begin{align}
\Pro(G^+) 
> \frac{\delta/\epsilon}{3/4 + \delta/\epsilon} 
> \frac{\delta}{\epsilon + \delta}.  
\label{eq:ub-Pro-G-plus}
\end{align}

According to the assumption at the beginning of this proof, $\mathbb{E}(v \given g) \le 1 - \delta - \epsilon$ for all $g \in G^+$.
Consider the expected payoff $\mathbb{E}(v)$ the mechanism yields for the principal. We have
\begin{align*}
\mathbb{E}(v)
&= \sum_{g \in G^+} \Pro(g) \cdot \mathbb{E}(v \given g) + \sum_{g \in G^-} \Pro(g) \cdot \mathbb{E}(v \given g) \\
&< (1 - \delta - \epsilon) \cdot \sum_{g \in G^+} \Pro(g)  + \sum_{g \in G^-} \Pro(g) \cdot 1 \\
&= (1 - \delta - \epsilon) \cdot \Pro(G^+) + (1 - \Pro(G^+)) \\
&= 1 - (\delta + \epsilon) \cdot \Pro(G^+) \\
&< 1 - \delta, 
\end{align*}
where we use the fact that the principal's maximum attainable payoff is $1$ in \Cref{exp:EnSESA-gt-SESA} to bound $\mathbb{E}(v \given g)$ for $g \in G^-$, and the last transition follows by \Cref{eq:ub-Pro-G-plus}.
This contradicts the assumption that $\mathbb{E}(v) \ge 1 - \delta$.
\end{proof}

\subsection{Proof of \Cref{lmm:g-alpha}}

\lmmgalpha*

\begin{proof}
Given \Cref{lmm:u-g}, we can find a signal $g$ such that
\begin{align}
\mathbb{E}(v \given g) &> 1 - \delta - \epsilon 
\label{eq:u-g-gt} \\
\text{and}\quad
\Pro(\gamma \given g) &\ge 1/4 - \delta/\epsilon.
\label{eq:pro-theta-given-g-ge}
\end{align}
Specifically, choosing 
$\epsilon = 10 \delta$ (recall that we have $M > 10000$ and $\delta = 10/M$, so $\epsilon < 1$), we get that
\begin{align}
\mathbb{E}(v \given g) &> 1 - 11\delta 
\label{eq:u-g-gt-11-delta} \\
\text{and}\quad
\Pro(\gamma \given g) &> 1/10.
\label{eq:pro-theta-given-g-ge-1-10}
\end{align}

First, note that since the principal gets payoff $1$ only when the agent plays $a'$ or $b'$, and $0$ otherwise, it must be that
\[
\Pro[t_1,\ \gamma,\ \neg (a' \vee b')\ \given\ g] < \delta + \epsilon,
\]
i.e., the probability that a type-$t_1$ agent plays an action outside of $\{a', b'\}$ when the state is $\gamma$ cannot exceed $\delta + \epsilon$.
Indeed, if this inequality does not hold, we would have $\mathbb{E}(v \given g) \le 1 - \delta - \epsilon = 1 - 11 \delta$, contradicting \Cref{eq:u-g-gt}.
It follows that
\begin{align}
\Pro(t_1 \given g) \cdot 
\Pro[\gamma,\ \neg(a' \vee b')\ \given\ t_1,\ g] < \delta + \epsilon
\label{eq:pro-t1-theta-given-g}
\end{align}
Since the occurrence of type $t_1$ is independent of the occurrence of $g$, we have $\Pro(t_1 \given g) = \Pro(t_1) = \rho(t_1) = 1/3$. Hence, 
\[
\Pro[\gamma,\ \neg(a' \vee b')\ \given\ t_1,\ g] < 3 (\delta + \epsilon),
\]
and in turn
\begin{align}
\Pro(\gamma,\ a' \vee b' \given t_1,\ g)
&> \Pro(\gamma \given t_1, g) -  3 (\delta + \epsilon) \nonumber \\
&> 1/10 -  3 (\delta + \epsilon),
\label{eq:ub-pro-ap-bp-cp-given-t1-theta-g}
\end{align}
where we used the fact that $\Pro(\gamma \given t_1, g) = \Pro(\gamma \given g)$ due to the independence between $\gamma$ and $t_1$ conditioned on $g$.

We now turn to the following upper bound of a type-$t_1$ agent's expected payoff,
conditioned on them receiving $g$ at the first step:
\begin{align*}
\mathbb{E} \left[ u_{t_1} \given t_1, g\right]
\ \le\ 
\Pro(\gamma, b' \given t_1, g) \cdot (-M) 
\ +\ 
(1 - \Pro(\gamma, b' \given t_1, g)) \cdot 3,
\end{align*}
where $-M = u_{t_1}(b', \gamma) = u_{t_1}(c', \gamma)$ and $3$ is the maximum attainable payoff for any possible combination of action and state,
according to the agent's payoffs in \Cref{exp:EnSESA-gt-SESA}.
It must be that $\mathbb{E} \left[ u_{t_1} \given t_1, g\right] \ge 0$ because otherwise the type-$t_1$ agent would be strictly better off if they stick to action $d$ after receiving $g$, whereby a payoff of $0$ is guaranteed.
Hence, we have  
\[
\Pro(\gamma, b' \given t_1, g) \cdot (-M) + (1 - \Pro(\gamma, b' \given t_1, g)) \cdot 3 \ge 0,
\]
which implies
\begin{align*}
\Pro(\gamma,\ b' \given t_1, g) \le \frac{3}{M+1}.
\end{align*}
Since playing $a'$ and playing $b'$ are mutually exclusive events, it follows by this inequality and \Cref{eq:ub-pro-ap-bp-cp-given-t1-theta-g} that
\begin{align}
\Pro(\gamma, a' \given t_1, g) > \frac{1}{10} -  3 (\delta + \epsilon) - \frac{3}{M+1}.
\label{eq:ub-pro-theta-ap-given-t1-g}
\end{align}

Similarly to the upper bound of $\mathbb{E} [ u_{t_1} \given t_1, g]$, the following upper bound can be established for a type-$t_1$ agent's expected payoff $\mathbb{E} \left[ u_{t_1} \given a', t_1, g\right]$ conditioned on $g$ and $a'$ being played:
\begin{align*}
\mathbb{E} \left[ u_{t_1} \given a', t_1, g\right]
&\ =\ 
\Pro(\alpha \given a', t_1, g) \cdot 1 
\ +\ \Pro(\gamma \given a', t_1, g) \cdot (-1) 
\ +\ \Pro(\beta \given a', t_1, g) \cdot (-M) \\
&\ \le\
\Pro(\alpha \given a', t_1, g) 
\ -\ \Pro(\gamma \given a', t_1, g).
\end{align*}
According to \Cref{eq:ub-pro-theta-ap-given-t1-g}, $\Pro(\gamma, a' \given t_1, g) > 0$, which means 
action $a'$ is indeed played by a type-$t_1$ agent with a positive probability conditioned on $g$.
It must then be that $\mathbb{E} \left[ u_{t_1} \given a', t_1, g\right] \ge 0$: if the type-$t_1$ agent were to expect a negative payoff for playing $a'$, then they would be better off deviating from $a'$ to $d$. Note in particular that when $a'$ is played, it must be that the type-$t_1$ agent reported their type truthfully because $a'$ is strictly dominated by $a$ or $d$ for types $t_2$ for any posterior distribution over the states---$a'$ would not be played if $t_2$ were reported. Hence, the above described deviation will not be blocked by the credibility constraint.
As a result, 
$\Pro(\alpha \given a', t_1, g) \ge \Pro(\gamma \given a', t_1, g)$, which implies
$\Pro(\alpha, a' \given t_1, g) \ge \Pro(\gamma, a' \given t_1, g)$ by multiplying the two sides with $\Pro(a' \given t_1, g)$.
It follows that 
\begin{align*}
\Pro(\alpha \given g) 
&\ge \Pro(t_1 \given g) \cdot \Pro(\alpha \given t_1, g) \\
&= \rho(t_1) \cdot \Pro(\alpha \given t_1, g) \\
&\ge \rho(t_1) \cdot \Pro(\alpha, a' \given t_1, g) \\
&\ge \rho(t_1) \cdot \Pro(\gamma, a' \given t_1, g) \\
&\ge \frac{1}{30} -  (\delta + \epsilon) - \frac{2}{3(M+1)} & \text{(by \Cref{eq:ub-pro-theta-ap-given-t1-g})} \\
&> \frac{1}{100}. & \text{(as $M > 10000$, $\delta = 10/M$, and $\epsilon = 10 \delta$)}
\end{align*}

By symmetry the same augment also applies to deriving: $\Pro(\beta \given g) > 1/100$ by changing $t_1$ to $t_2$ and $\alpha$ to $\beta$ in the proof.
\end{proof}

\subsection{Proof of \Cref{lmm:EnSEA-optimal}}

\lmmEnSEAoptimal*

\begin{proof}
Given any $\Pi$, we will construct an $\EnSES$ mechanism $\Pi'$ and show that the conversion ensures the stated properties while preserving the principal's payoff.
Without loss of generality, we assume that there is exactly one $\xE$ node on every path from the root to a leaf. Indeed, there is no need to elicit the agent's type more than once due to the credibility requirement. On the other hand, if there is no $\xE$ on some path, we can replace the leaf node on this path with an $\xE$ node (whose children are all leaf nodes). 
Moreover, we assume that every path begins and ends with signaling stages. This can be achieved by adding a dummy $\xS$ node as the parent of the root node of $\Pi$ where the same signal is sent irrespective of the states, and similarly by inserting a dummy $\xS$ node before every leaf node.

\medskip

Next, we apply the following operations to construct $\Pi'$.

\paragraph{Step 1. Adding a New root}

First, we create an $\xEN$ node as the root node of $\Pi'$.
Then we make $n$ copies $\Pi^1,\dots,\Pi^n$ of $\Pi$ and attach each $\Pi^i$ to the root of $\Pi'$ as a subtree. 
We mark the root of each $\Pi^i$ as $\xS_i$.
It is straightforward that $\Pi'$ constructed so far is IC at the root node: it is optimal for each type $t_i$ to select the $i$-th child $\Pi^i$ of $\Pi'$. 
Since all the copies are the same, each type $t_i$ is incentivized to choose $\Pi^i$, so $\Pi'$ is IC at the root.
Moreover, the payoffs $\Pi'$ yields for both the principal and every agent type are the same as those yielded by $\Pi$.

\paragraph{Step 2. Removing $\xEN$ nodes}

Given IC at the root, in each branch $\Pi^i$, we only need to consider the response of $t_i$.
Hence, all the $\xEN$ nodes in $\Pi^i$ become redundant.
For each $\xEN$ node $e$ in $\Pi^i$, we remove $e$, connecting the branch of $e$ a type-$t_i$ agent is incentivized to select directly to the parent of $e$. (All other branches of $e$ are removed.)
It is straightforward that this operation preserves the payoff of each type $t_i$ for $\Pi^i$, while it can only decrease their payoff for each $\Pi^j$, $j \neq i$, now that there are fewer choices in each $\Pi^j$.
Hence, $\Pi'$ remains IC and yields the same payoffs for the players as before.

\paragraph{Step 3. Merging $\xS$ nodes}

Next, we merge consecutive $\xS$ nodes on every path.
If any $\xS$ node $s$ has a child $s'$ which is also an $\xS$ node, we remove the edge connecting $s$ and $s'$ and attach the children of $s'$ directly to $s$. We also modify the signaling strategy $\pi_s$ at $s$: replace the signal $g$ in $\pi_s$ that corresponds to $s'$ with signals in the support of $\pi_{s'}$, and redefine $\pi_s( g \given \theta) = 0$ and
\[
\pi_s( g' \given \theta) = \pi_{s'}( g \given \theta) \cdot \pi_{s'}( g' \given \theta)
\]
for each signal $g'$ in the support of $\pi_{s'}$.
Namely, the signals sent at $s'$ is now generated directly at $s$. The modified probabilities in $\pi_s$ ensures that the posterior distributions induced by these signals remain the same as before.

\medskip

The above operations yield a five-level tree. Levels 1 to 5 contain only $\xEN$, $\xS$, $\xE$, $\xS$, and leaf nodes, respectively. 
The mechanism is IC at $\xEN$, but may not yet be DIC at the other $\xS$ nodes. We proceed with the following steps to ensure the DIC properties.

\paragraph{Step 4. Ensuring DIC at $\xS_i$ and $\xS_{ijk}$}

Recall that $\xS_i$ is the root of $\Pi^i$ and now at the second level of $\Pi'$, and $\xS_{ijk}$ is at the fourth level.
To ensure DIC, we merge signals sent at these nodes that induce the same agent response (i.e., reporting the same type or taking the same action).

Specifically, consider an arbitrary $\xS_i$ and let us first identify the following sets of nodes:
\begin{itemize}
\item 
For each $j=1,\dots,n$, let $\mathcal{E}_{ij}$ be the set consisting of the children of $\xS_i$ (which are all $\xE$ nodes) at which a type-$t_i$ agent is incentivized to report $t_j$.

\item 
For each $k=1,\dots,n$, we let $\mathcal{S}_{ijk}$ be the set consisting of the children of $\mathcal{E}_{ij}$ that correspond to the agent reporting type $t_k$. 
\end{itemize}
We merge nodes in each $\mathcal{E}_{ij}$ into a single $\xE$ node.
In other words, we replace signals leading to these nodes with a single signal, which can be viewed as a signal that recommends the agent to report $t_j$. 
This operation also makes the nodes in each $\mathcal{S}_{ijk}$ indistinguishable for the agent. 
In other words, these nodes are now in the same information set. 
We treat each information set as a new node and denote them by $\xE_{ij}$ or $\xS_{ijk}$, accordingly. 
The signaling strategies at $\xS_i$ and $\xS_{ijk}$ are updated accordingly as follows.
\begin{itemize}
\item 
We update the strategy $\pi_i$ at $\xS_i$ by merging the probabilities of the merged signals: for every $\theta \in \Theta$,
\[
\pi_i(g_{\xE_{ij}} \given \theta) 
\leftarrow \sum_{e \in \mathcal{E}_{ij}} \pi_{i} (g_{e} \given \theta)
\]
where $g_e$ denotes the signal leading to a node $e$
(and we abandon all the original signals $g_e$ for all $e \in \mathcal{E}_{ij}$).

\item 
The signaling strategy $\pi_{ijk}$ at each $\xS_{ijk}$ is defined over $\bigcup_{e \in \mathcal{S}_{ijk}} G_e$, where $G_e$ denotes the signal space of the strategy at $e$. (W.l.o.g., the $G_e$'s are disjoint sets because in $\Pi$ the agent always knows which node $e$ they are at.)
For every $g \in \bigcup_{e \in \mathcal{S}_{ijk}} G_e$, we let 
\[
\pi_{ijk}(g \given \theta) \leftarrow \sum_{e \in \mathcal{E}_{ij}}
w_{e} \cdot \pi_{ek}(g \given \theta),
\]
where $ek$ denotes the child of $e$ corresponding to type $t_k$, and $\pi_{ek}$ is the signaling strategy at $ek$; the normalization factor
$w_{e} := \frac{ \pi_{i} ( g_e) } {\sum_{e' \in \mathcal{E}_{ij}} \pi_{i} ( g_{e'})}$ is the probability of $e$ conditioned on $\mathcal{E}_{ij}$. 
\end{itemize}
The above design ensures that each $g$ sent by $\pi_{ijk}$ induces the same posterior belief as it previously does. 
Hence, at the new node $\xE_{ij}$, a type-$t_i$ agent would still be incentivized to report $t_j$, so DIC is ensured at $\xS_i$.

DIC at each $\xS_{ijk}$ can be ensured similarly by merging signals that induce the same action of the agent.
Essentially, the modifications in Step~4 make many nodes in $\Pi^i$ indistinguishable to the agent, but if the agent follows the recommendations in the new mechanism, they would still reach each leaf with the same probability as before, as if they respond optimally in $\Pi^i$.
On the other hand, since these merging operations reduce information, the agent cannot devise any better response to improve their payoff. Hence, the agent is incentivized to follow the recommendations and this ensures DIC. 
\end{proof}

\subsection{Proof of \Cref{lmm:leaf-T-x}}

\lmmleafTx*


\begin{proof}
Since $\Pi$ is optimal, it yields payoff $1$ for the principal. This means that only action $a'_i$ can be played in state $\theta_i$, so no action will be played at two different states both with positive probabilities. Now that $x$ is a leaf node, the agent needs to choose an action to play. However, since $\Theta_x$ contains more than one state, no matter which action the agent chooses, the same action will be played at more than one state, a contradiction.
\end{proof}

\subsection{Proof of \Cref{lmm:Theta-x-T-x}}

\lmmThetaxTx*


\begin{proof}
Suppose that $\theta_i \in \Theta_x$ but $t_i \notin T_x$. Since $\Pi$ is optimal, it yields payoff $1$ for the principal. 
This means that action $a'_i$ must be played in state $\theta_i$, so we have $\Pro(a'_i \given x) > 0$. 
Nevertheless, $a'_i$ is strictly dominated for all types in $T \setminus \{t_i\}$, so it will never be played if the agent imitates any type in $T_x$, now that $t_i \notin T_x$. Hence, there is a contradiction.
\end{proof}

\subsection{Proof of \Cref{lmm:child-in-Pi}}

We first prove the following lemma.

\begin{lemma}
\label{lmm:ub-p-theta-n}
For every node $x$ in $\Pi$, 
if $t_i \in T_x$,
then $p_x(\theta_j) \le p_x(\theta_n) + 2/M$ for all $j = i+1,\dots,n$.
\end{lemma}

\begin{proof}
Suppose for the sake of contradiction that for some $j \in \{i+1, \dots, n\}$ it holds that 
\begin{equation}
\label{eq:px-j-gt-px-n}
p_x(\theta_j) > p_x(\theta_n) + 2/M
\end{equation}
Since $\Pi$ is optimal, according to our analysis in \Cref{sec:opt-pie}, it must yield payoff $1$ for the principal.
Hence, only the outcomes $(\theta_i, a'_i)$, $i=1,\dots,n$, have positive probabilities in the probability measure induced by $\Pi$. In this case, a type-$t_0$ agent's expected payoff conditioned on $x$ is $\sum_{i=1}^n p_x(\theta_i) \cdot u_{t_0}(\theta_i, a'_i) = p_x(\theta_n)$. 
We demonstrate the following contradiction: there is a way for a type-$t_0$ agent to obtain strictly more than $p_x(\theta_n)$.

Consider the outcomes induced if a type-$t_0$ agent imitates to $t_i$ in the subtree starting at $x$; this is feasible given that $t_i \in T_x$. Namely, the agent always selects a subset containing $t_i$ in response to elicitation in the subtree and eventually plays an optimal action of $t_i$. 
Let $\Pro_*$ be the probability measure induced by $\Pi$ and this strategy of the agent. 
According to the incentive of type $t_i$, 
if $\Pro_*( a'_i \given x) > 0$, then
it must be that 
\begin{equation}
\label{eq:ub-neg-theta-i}
\Pro_*(\neg\theta_i \given a'_i, x) < 1/M^2.
\end{equation}
Indeed, if $\Pro_*(\neg\theta_i \given a'_i, x) \ge 1/M^2$, type $t_i$'s expected payoff conditioned on $a'_i$ being played in the subtree would be {\em at most}
\[
\Pro_*(\neg\theta_i \given a'_i, x) \cdot (-M) + \left(1 - \Pro_*(\neg\theta_i \given a'_i, x) \right) \cdot \underbrace{1/M}_{u_{t_i}(\theta_i, a'_i)} < 0.
\]
In this case, type $t_i$ would be overall better off swapping action $a'_i$ with $a_i$ whenever the former is to be played, contradicting the assumption that $\Pro_*( a'_i \given x) > 0$.

Now consider the type-$t_0$ agent's actual payoff over $\Pro_*$, that is 
$\mathbb{E}_*(u_{t_0} \given x) 
= \sum_{\theta \in \Theta} 
\sum_{a \in A} \Pro_*(\theta, a \given x) \cdot u_{t_0}(\theta, a)$.
We can exclude actions in $\{a'_1, \dots, a'_n\} \setminus \{a'_i\}$, which are strictly dominated by other actions according to type $t_i$'s incentive. It follows that
\begin{align}
\mathbb{E}_*(u_{t_0} \given x)
&= 
\sum_{\theta \in \Theta} \sum_{k=1}^n \Pro_*(\theta, a_k \given x) \cdot u_{t_0}(\theta, a_k)
\quad+\quad 
\sum_{\theta \in \Theta} \Pro_*(\theta, a'_i \given x) \cdot u_{t_0}(\theta, a'_i) \nonumber \\
&\ge 
\sum_{\theta \in \Theta} \sum_{k=1}^n \Pro_*(\theta, a_k \given x) \cdot \underbrace{u_{t_i}(\theta, a_k)}_{\hspace{-5mm}\text{as }u_{t_0}(\theta, a_k) \ge u_{t_i}(\theta, a_k)\hspace{-5mm}}
\qquad+\quad 
\underbrace{\Pro_*(\neg\theta_i, a'_i \given x) \cdot (-M)}_{\text{as }u_{t_0}(\theta_i, a'_i) = 0 \text{ and }u_{t_0}(\theta, a'_i) = -M\, \forall \theta \neq \theta_i} 
\label{eq:E-u-t0-x}
\end{align}
The second term in the last line can be bounded from below by $-1/M$ because we have 
\[
\Pro_*(\neg\theta_i, a'_i \given x) 
= \Pro_*(a'_i \given x) \cdot \Pro_*(\neg\theta_i \given a'_i, x) 
\le \Pro_*(\neg\theta_i \given a'_i, x) < 1/M^2
\]
where $\Pro_*(\neg\theta_i \given a'_i, x) < 1/M^2$ is implied by \Cref{eq:ub-neg-theta-i}. 
We next bound the first term.
By definition, $\Pro_*$ is induced by a type-$t_i$ agent's best response. So, the first term can be bounded from below by what a type-$t_i$ agent would get if he swapped all the actions $a_1, \dots, a_n$ to $a_j$. (Recall $j$ is the number that makes \Cref{eq:px-j-gt-px-n} hold.)
Specifically, we have
\begin{align*}
\sum_{\theta \in \Theta} \sum_{k=1}^n \Pro_*(\theta, a_k \given x) \cdot u_{t_i}(\theta, a_j)
& = 
\sum_{k=1}^n \Pro_*(\theta_j, a_k \given x) \cdot \underbrace{u_{t_i}(\theta_j, a_j)}_{=1} +  \sum_{\theta \neq \theta_j} \sum_{k=1}^n \Pro_*(\theta, a_k \given x) \cdot \underbrace{u_{t_i}(\theta, a_j)}_{=0} \\
&=
\sum_{k=1}^n \Pro_*(\theta_j, a_k \given x) \\
&=
\Pro_*(\theta_j \given x) - 
\Pro_*(\theta_j, a'_i \given x) - \sum_{k \neq i}\ \underbrace{\Pro_*(\theta_j, a'_k \given x)}_{=0 \text{ as } a'_k \text{ is strictly dominated}} \\
&=
\Pro_*(\theta_j \given x) - 
\Pro_*(\theta_j, a'_i \given x) \\
&\ge 
\Pro_*(\theta_j \given x) - 
\Pro_*(\theta_j \given a'_i, x) 
>
\Pro_*(\theta_n \given x) + 2/M - 1/M^2,
\end{align*}
where the last transition is due to \Cref{eq:px-j-gt-px-n,eq:ub-neg-theta-i}.
Substituting into \Cref{eq:E-u-t0-x} this lower bound and the bound $-1/M$ we derived for the second term, we get the contradiction
$\mathbb{E}_*(u_{t_0} \given x) > p_x(\theta_n)$, as desired.
\end{proof}

\lmmchildinPi*

\begin{proof}
If $x$ is an $\xE$ node, the posteriors at its children would be the same as that at $x$. The statement holds trivially. Hence, in what follows we consider the case where $x$ is an $\xS$ node.

Suppose for the sake of contradiction that $t_i \in T_x$ and $p_x \in \calP_i^{\epsilon_k}$, but $p_c \notin \calP_i^{\epsilon_{k+1}}$ for all $c \in \children(x)$.
Note that condition (i) holds automatically for all $c \in \children(x)$ because $p_c(\theta_{j'})$ remains $0$ for all $j < i$.
Hence, every $c \in \children(x)$ either violates (ii) $\sum_{j = i+1}^n p_c(\theta_j) \ge 1 - \frac{1}{n-i+1} - \epsilon_{k+1}$, or 
(iii) $\left| p_c(\theta_j) - p_c(\theta_n) \right| \le \epsilon_{k+1}/2$ for some $j$.

We can then define a partition $\{C_0, C_1, \dots, C_n\}$ of $\children(x)$, where $C_0$ contains children violating (ii), and each other $C_j$ contains children violating (iii) for $j$. (If a child violates multiple conditions, we assign it to an arbitrary set.)

We next bound the conditional probabilities $\Pro(C_j \given x) = \sum_{c \in C_j} \Pro(c \given x)$ of these subsets, where $\Pro$ is the probability measure induced by $\Pi$. 
Note that for every $j$ we can write $p_x(\theta_j)$ as the weighted average
$p_x(\theta_j) =
\sum_{c \in \children(x)} \Pro(c \given x) \cdot p_c(\theta_j)$.
\begin{itemize}
\item 
Every $c \in C_0$ violates the first condition, so $\sum_{j = i+1}^n p_c(\theta_j) < 1 - \frac{1}{n-i+1} - \epsilon_{k+1}$.
It must be that 
\begin{align}
\label{eq:C0-ub}
\Pro(C_0 \given x) \le \xi := 
\frac{\frac{1}{n-i+1} + \epsilon_k}{\frac{1}{n-i+1} + \epsilon_{k+1}}.
\end{align}
Indeed, if $\Pro(C_0 \given x) > \xi$ was true, we would have
\begin{align*}
\sum_{j=i+1}^n p_x(\theta_j) 
&=
\sum_{j=i+1}^n \left( \sum_{c \in C_0} \Pro(c \given x) \cdot p_c(\theta_j) \quad+\  \sum_{c \in \children(x) \setminus C_0} \Pro(c \given x) \cdot p_c(\theta_j) \right) \\
&=
\sum_{c \in C_0} \Pro(c \given x) \sum_{j=i+1}^n p_c(\theta_j) \quad+\  \sum_{c \in \children(x) \setminus C_0} \Pro(c \given x) \sum_{j=i+1}^n p_c(\theta_j) \\
&< 
\Pro(C_0 \given x) \cdot \left( 1 - \frac{1}{n-i+1} 
- \epsilon_{k+1} \right)  \quad+\quad \left(1 - \Pro(C_0 \given x) \right) \cdot 1 \\
&\le 
\xi \cdot \left( 1 - \frac{1}{n-i+1} 
- \epsilon_{k+1} \right)  \quad+\quad (1 - \xi) \qquad\qquad \text{(assuming $\Pro(C_0 \given x) > \xi$)}\\
&= 1 - \frac{1}{n-i+1} - \epsilon_{k},
\end{align*}
contradicting the assumption that $p_x \in \calP_i^{\epsilon_k}$.

\item 
For every $j = i+1, \dots, n$, every $c \in C_j$ violates the second condition for $j$, so
$\left| p_c(\theta_j) - p_c(\theta_n) \right| > \epsilon_{k+1}/2$.
By assumption $t_i \in T_x$.
Now that $x$ is an $\xS$ node, no type will be excluded from $T_x$ at $x$, so we have $T_c = T_x$ for all $c \in \children(x)$. 
It then follows by \Cref{lmm:ub-p-theta-n} that: 
\begin{align}
\label{eq:pc-theta-ell}
p_c(\theta_j) \le p_c(\theta_n) + 2/M.
\end{align}
Since $2/M < \epsilon_{k+1}/2$ is sufficiently small, $\left| p_c(\theta_j) - p_c(\theta_n) \right| > \epsilon_{k+1}/2$ holds only if
\begin{align}
\label{eq:pc-theta-j}
p_c(\theta_j) < p_c(\theta_n) - \epsilon_{k+1}/2.
\end{align}
Then, it must be that
\begin{equation}
\label{eq:Cj-ub}
\Pro(C_j \given x) \le \xi' := 2 \epsilon_k / \epsilon_{k+1},
\end{equation}
as the converse, $\Pro(C_j \given x) > \xi'$, would result in
\begin{align*}
p_x(\theta_j) 
&=
\sum_{c \in C_j} \Pro(c \given x) \cdot p_c(\theta_j) + \sum_{c \in \children(x) \setminus C_j} \Pro(c \given x) \cdot p_c(\theta_j) \\
&\le
\sum_{c \in C_j} \Pro(c \given x) \cdot \underbrace{\left( p_c(\theta_n) - \epsilon_{k+1}/2 \right)}_{\text{by \Cref{eq:pc-theta-j}}} + \sum_{c \in \children(x) \setminus C_j} \Pro(c \given x) \cdot \underbrace{\left( p_c(\theta_n) + 2/M \right)}_{\text{by \Cref{eq:pc-theta-ell}}} \\
&= 
\sum_{c \in \children(x)} \Pro(c \given x) \cdot p_c(\theta_n) 
\ -\ 
\Pro(C_j \given x) \cdot \epsilon_{k+1}/2
\ +\ 
\left( 1 - \Pro(C_j \given x) \right) \cdot 2/M \\
&< p_x(\theta_n) - \epsilon_{k} + 2/M \\
&< p_x(\theta_n) - \epsilon_{k}/2,
\end{align*}
contradicting $p_x \in \calP_i^{\epsilon_k}$, too.
\end{itemize}

Summing up \Cref{eq:C0-ub,eq:Cj-ub} gives the following contradiction, for $n$ sufficiently large: 
\begin{align*}
\sum_{j=0}^n \Pro(C_j \given x) 
\le
\xi + n \xi'
&=
\frac{\frac{1}{n-i+1} + \epsilon_k}{\frac{1}{n-i+1} + \epsilon_{k+1}} + n \cdot 2 \epsilon_k / \epsilon_{k+1}  \\
&\le
\frac{1 + \epsilon_k}{1 + \epsilon_{k+1}} + n \cdot 2 \epsilon_k / \epsilon_{k+1} 
\le 
\frac{\epsilon_{k+1} + n^2  \cdot \epsilon_k}{\epsilon_{k+1} + \epsilon_{k+1}^2} 
< 1,
\end{align*}
where we used the fact that $\epsilon_{k+1}^2 > n^2 \epsilon_k$ (recall that by definition $\epsilon_k = n^{-3^{n-k}}$ for $k \in \mathbb{Z}$). 
\end{proof}

\end{document}